\documentclass[showpacs,preprintnumbers,amsmath,amssymb,12pt]{article}
\usepackage{a4wide,amsmath,amsthm,epsfig,graphicx}
\usepackage{amsmath,amsthm,amssymb}
\usepackage{amsfonts}
\usepackage{graphics}
\usepackage{dsfont}
\usepackage{bbm}
\usepackage{bm}
\usepackage{xcolor}
\usepackage{dcolumn}
\usepackage{fancyhdr,fancybox}
\usepackage[svgnames]{pstricks}
\usepackage{pst-text}
\usepackage{pst-plot,pst-eucl,pstricks-add}
\usepackage{caption}
\usepackage{subcaption}
\usepackage{systeme, mathtools}
\usepackage{algorithm}
\usepackage{algpseudocode}
\usepackage{slashbox}
\usepackage{placeins}

\algrenewcommand\alglinenumber[1]{{\sffamily\footnotesize#1}}
\algnewcommand\Not{\textbf{not}}
\algnewcommand\FALSE{FALSE}
\algnewcommand\TRUE{TRUE}

\newtheorem{thm}{Theorem}[section]
\newtheorem{cor}[thm]{Corollary}

\newtheorem{prop}[thm]{Proposition}

\newtheorem{remark}{Remark}

\newcommand{\refeqq}[1]{~(\ref{#1})}
\newcommand{\myref}[1]{~\ref{#1}}
\newcommand{\mycite}[1]{~\cite{#1}}

\newcommand{\R}{\mathbb{R}}

\newcommand{\rv}{\textit{rv}}
\newcommand{\id}{\textit{id}}
\newcommand{\iid}{\textit{iid}}
\newcommand{\sd}{\textit{sd}}

\newcommand{\pdf}{\textit{pdf}}

\newcommand{\chf}{\textit{chf}}
\newcommand{\che}{\textit{che}}
\newcommand{\cgf}{\textit{cgf}}

\newcommand{\Levy}{L\'{e}vy}

\newcommand{\Pqo}{\bm{P}\hbox{-\emph{a.s.}}}
\newcommand{\eqd}{\stackrel{d}{=}}

\newcommand{\EXP}[1]{\bm{E}\left[{#1}\right]}
\newcommand{\VAR}[1]{\bm{V}\left[{#1}\right]}

\newcommand{\lch}{\textit{lch}}

\newcommand{\norm}{\mathcal{N}}

\newcommand{\poiss}{\mathcal{P}}

\newcommand{\unif}{\mathcal{U}}

\newcommand{\gam}{\mathcal{G}}

\newcommand{\arem}{$a$-remainder}

\newcommand{\cts}{\mathcal{TS}}

\newcommand{\sncts}{\mathcal{SNTS}}

\usepackage[a4 paper, top = 1.3 cm, bottom = 1.5 cm, left = 1.3 cm, right = 1.3 cm]{geometry}
\title{\Huge \textbf{Normal Tempered Stable Processes and the Pricing of Energy Derivatives
}}

\author{Piergiacomo \textsc{Sabino} \footnote{piergiacomo.sabino@eon.com
\newline
The views, opinions, positions or strategies expressed in this article are those of the author
and do not necessarily represent the views, opinions, positions or strategies of, and should not be
attributed to E.ON SE.}
\\
Quantitative Risk Management\\
E.ON SE\\
\vspace{5pt}
 Br\"{u}sseler Platz 1, 45131 Essen, Germany
}

\date{}
\begin{document}
    \maketitle \thispagestyle{empty}
        \begin{abstract}
\noindent				In this study we consider the pricing of energy derivatives when the evolution of spot prices is modeled with a normal tempered stable driven Ornstein-Uhlenbeck process. Such processes are the generalization of normal inverse Gaussian processes that are widely used in energy finance applications. We first specify their statistical properties calculating their characteristic function in closed form.  
This result is instrumental for the derivation of non-arbitrage conditions such that the spot dynamics is consistent with the forward curve without relying on numerical approximations or on numerical integration. Moreover, we conceive an efficient algorithm for the exact generation of the trajectories which gives the possibility to implement Monte Carlo simulations without approximations or bias. 
We illustrate the applicability of the theoretical findings and the simulation algorithms in the context of the pricing of different contracts, namely, strips of daily call options, Asian options with European style and swing options. 
Finally, we present an extension to future markets.
\\
				
        \end{abstract}


        \section{Introduction}\label{sec:introduction}
				
Energy and commodity markets exhibit mean-reversion, seasonality and sudden spikes; mean-reversion in particular, cannot be captured by ordinary \Levy\ processes. To this end, the modeling based on non-Gaussian Ornstein-Uhlenbeck (OU) processes has received considerable attention in the recent literature in an attempt to accommodate features such as jumps, heavy tails which are well evident in real data. 
				
A common approach to describe the evolution of day-ahead (spot) prices assumes the following multi-factor dynamics 
\begin{equation*}
S(t) = F(0,t)\,\exp\left(h(t) + X(t) \right) = F(0,t)\,\exp\left(h(t) + \sum_n^NX_n(t) + \sum_{m}^L Y_m(t) \right).
\end{equation*}				
$F(0,t)$ represents the forward curve derived from quoted products and reflects the seasonality, $X_n(\cdot), n=1, \dots$ are independent mean-reverting \Levy-driven OU processes and $Y_m(\cdot), m=1, \dots$ are independent plain \Levy\ processes. Due to the Lemma 3.1 in Hambly
et al.\mycite{HHM11},  the risk-neutral conditions are met when the deterministic function $h(t)$ is consistent with the forward curve such that
\begin{equation*}
	h(t) = -m_X(1, t)
\end{equation*}
where $m_X(s, t)$ is the logarithm of the moment generating function of $X(t)$. Standard examples in the Gaussian framework are the one factor model ($N=1$, $M=0$) of Lucia and Schwarz\mycite{LS02} and the two factor version ($M=N=1$) of Schwartz and Smith\mycite{SchwSchm00}. 
The extension to OU-based models driven by a normal inverse Gaussian (NIG)
process or by a variance gamma process can be found in Benth et al\mycite{BKM07}, Cummins et al.\mycite{CKM17, CKM18} and Sabino\mycite{Sabino20b}. Finally, a third type of market models is based on mean-reverting jump-diffusion OU processes as originally suggested in Cartea and Figueroa\mycite{CarteaFigueroa} and further analyzed for instance in Kjaer\mycite{Kjaer2008}, Humbly et al\mycite{HHK09} and recently in Sabino and Cufaro Petroni\mycite{cs20_2}.

On the other hand, in order to select a realistic and viable model for derivative pricing, one has to first answer the following questions.
	\begin{itemize}
		\item Do we know the characteristic function and the moment generating function of $X(t)$ to determine the non-arbitrage conditions.
		\item Do we know anything about the statistics of the process $X(\cdot)$? How viable is the parameters estimation?
		\item Do we know how to generate the trajectories of $X(\cdot)$ in order to implement Monte Carlo simulations?
	\end{itemize}
The availability of simulation techniques of easy implementation is important for analysis,
validation and estimation purposes. Indeed, direct likelihood analysis is often
impracticable, whereas Monte Carlo (MC) based techniques and generalized method of moments (GMM) approaches can be a viable route to estimate the model parameters.				

In this paper we study normal tempered stable (NTS) processes, which generalize NIG processes, and focus on their OU counterpart. Benth and Benth\mycite{BenthBenth04}, Benth et al.\mycite{BKM07, BDPL18}  illustrate the suitability of NIG-driven OU models in various contexts, namely in gas, oil, power markets and also in the pricing of wind derivatives. On the other hand, these applications rely on approximated solutions to the answers mentioned above that may produce biased results, for instance in the pricing of forward start contracts as we shall show.

We address each of the aforementioned questions for models that are based on symmetric normal tempered stable driven OU processes. The first contribution then is the derivation of the characteristic function of their transition law in closed form. This result is instrumental to determine their statistical properties and compared to the current state of affairs (see once again Benth and Benth\mycite{BenthBenth04} and Benth et al.\mycite{BKM07}) we can derive non-arbitrage conditions without resorting to a numerical approximation or a numerical integration. 

The second contribution consists in defining an efficient algorithm for the exact simulation of the trajectories of symmetric normal tempered stable (and therefore NIG) driven OU processes which is particularly suitable for forward start contracts where the standard Euler scheme would return biased results. Indeed, even though it is common practice to rely on such an approximated scheme, we show it can lead to biased results. We also propose an alternative approximation scheme that apparently outperforms the Euler approximation having the same computational effort.

We illustrate the applicability of our results to the pricing of three common derivative contracts in energy markets, namely a strip of daily call options, an Asian option with European style and a swing option.
In the first example we consider a one-factor model purely driven by a single mean-reverting OU-NTS process basically coinciding with that in Benth and Benth\mycite{BenthBenth04} and Benth et al.\mycite{BKM07}. We make use of the explicit knowledge of the characteristic function  to implement the pricing with the FFT-based technique of Carr and Madan\mycite{Carr1999OptionVU}, and then compare the outcomes to those obtained via MC simulations. We observe that such a model is not really suitable for the pricing of contracts having long maturities which in contrast require a second NTS factor. The resulting dynamics is a NTS version of the model of de Jong and Schneider\mycite{DJS09} in the Gaussian framework.

In the second example, we calibrate the two factor model taking the German day-ahead and month-ahead NCG prices and evaluate a forward start Asian option. Based on our theoretical results, we suggest to estimate the parameters using the alternative approximation, whereas to employ the exact simulation scheme for the pricing of energy derivatives especially if they are forward start contracts. 
In addition, we have shown that the proposed simulation algorithm, combined with the Least Squares Monte Carlo approach of  Hambly et al\mycite{HHK09}, provides an efficient and accurate pricing of a one year $120-120$ swing option.

Finally, our results are not only restricted to OU-processes and spot models. Indeed, we shall detail how they can be easily extended to forward markets in order to capture the Samuelson effect and different implied volatility profiles displayed by options futures in the spirit of the works of Benth et al.\mycite{BPV19}, Latini et al.\mycite{LPV19} and Piccirilli et al.\mycite{PSV20}.			
	The paper is organized as follows: Section\myref{sec:Preliminaries} introduces the notation and presents the mathematical preliminaries. 
	In Section\myref{sec:ou:ncts} we focus on the symmetric normal tempered stable processes of OU type, we derive the characteristic function of their transition law, and present the algorithms for their exact simulation. We also carry out numerical experiments demonstrating their efficiency. The financial application of
these results is illustrated in Section\myref{sec:fin:app} in the context of the pricing of energy derivative
contracts, namely daily strips of call options, Asian options with European exercise
and swing options written on the day-ahead spot price. Furthermore, we detail how our findings are not restricted to the modeling of the spot dynamics and can be easily extended to future markets.
Finally Section\myref{sec:conclusions} concludes
the paper with an overview of future investigations and possible further applications.
\section{Notations and preliminary remarks}\label{sec:Preliminaries}
In this section we present the concepts of normal variance-mean mixtures and of non-Gaussian Ornstein-Uhlenbeck processes that will be instrumental for the modeling of energy markets. We also introduce the notation and the shortcuts that will be used throughout the paper.

\subsection{Notation}\label{sub:sec:notation}
The function $\Gamma(x)$ represents the Euler gamma function, in addition, we write $\norm(\mu, \sigma^2)$ to denote the Gaussian distribution with mean $\mu$ and variance $\sigma^2$. Moreover, we write $\unif([0,1])$ to denote the uniform distribution in $[0,1]$ and $\poiss(\lambda)$ to denote the Poisson distribution with parameter $\lambda>0$.
We use the shortcuts \id\ and \sd\ for \emph{infinitely divisible} and \emph{self-decomposable} distributions, respectively. We use the shortcut \rv\ for \emph{random variable} and \iid\ for \emph{independently and identically distributed}, whereas we use \chf, \lch, \cgf\ and \pdf\ as shortcuts for \emph{characteristic function}, \emph{logarithmic characteristic}, \emph{cumulant generating function} and \emph{density function}, respectively.  


\subsection{Normal variance-mean mixtures}\label{sub:sec:normal:mixtures}
Consider a standard normal \rv\ $X\sim\norm(0,1)$ and an independent \rv\ $V$ with \pdf\ $f_V(x)$ defined on the positive real axis $\R^+$, then a \rv\ $Y$ is said to be distributed according to a \emph{normal variance-mean mixture} with mixture \pdf\ $f_V(x)$ if it has the form
\begin{equation}
Y = \mu + \theta\,V + \sigma\sqrt{V}\,X
\label{eq:nvmm}
\end{equation}
where $\theta$, $\mu$ and $\sigma>0$ are real numbers. The conditional distribution of $Y$ given $V$ is thus a Gaussian distribution $\norm(\mu + \theta V, \sigma^2\,V)$, whereas the unconditional \pdf\ and \chf\, denoted $f_Y(x)$ and $\varphi_Y(u)$ respectively, are 
\begin{equation}
f_Y(x) = \int_{0}^{\infty} \frac{1}{\sqrt{2\,\pi\,\sigma^2\,v}}
\exp\left(-\frac{\left(x - \mu - \theta\,v\right)^2}{2\,\sigma^2\,v}\right)dv
\label{eq:pdf:nvmm}
\end{equation} 
\begin{equation}
\varphi_Y(u) = e^{i\,u\,\mu}\varphi_V\left(\theta\,u + \frac{i\,\sigma^2\,u^2}{2}\right)
\label{eq:chf:nvmm}
\end{equation}
where $\varphi_V(u)$ denotes the \chf\ of $V$.

Assuming for simplicity $\mu=0$, the law of $Y$ coincides with the distribution of the position of a subordinated Brownian motion (BM) where the law of the subordinator at a fixed time is given by $f_V(x)$. We recall that a subordinator $L(\cdot)$ is a pure jump \Levy\ process with non-decreasing trajectories and \Levy\ measure $\nu_L$ such that 
\begin{equation*}
\int_0^{\infty}(x\wedge 1)\nu_L(dx)<\infty, \quad \nu_L((-\infty, 0)) = 0.
\end{equation*} 
Denoting $W(\cdot)$ the standard Wiener process, a subordinated BM  
\begin{equation}
Y(t) = \theta\,L(t) + \sigma\,W(L(t))
\label{eq:sub:bm}
\end{equation}
is thus a new \Levy\ process  with time-change given by $L(\cdot)$.

Furthermore, a \Levy\ process $L(\cdot)$ is said to be a classic tempered stable (dubbed TS) subordinator if its \Levy\ measure has density 
\begin{equation}
\ell_L(x) = c\,\frac{e^{-\beta\,x}}{x^{1+\alpha}}\mathds{1}_{x\ge 0}\label{eq:cts:density}
\end{equation}
where $c, \beta$ are all positive numbers, $0<\alpha<1$. Hereafter $\cts(\alpha, \beta, c\,t)$ denotes the law of such a subordinator at time $t$. This class is quite flexible and contains the  inverse Gaussian (IG) distribution that corresponds to $\alpha=1/2$ and the gamma distribution for the limiting case $\alpha=0$. 
For sake of completeness, we remark that we are referring to classical TS processes because alternative tempering functions rather than the exponential function can be used in\refeqq{eq:cts:density} (see  Rosi\'nski\mycite{ROSINSKI2007677}). An overview of such processes, named general tempered stable processes, can be found in Grabchak\mycite{Grabchak16}.

By time-changing a BM with drift $\theta$ and volatility $\sigma$ with a TS subordinator, we obtained the so-called \emph{normal (classical) tempered stable} process, hereafter dubbed NTS;
taking $\alpha=1/2$ we have the well-known NIG process. On the other hand, due to the scaling properties of the TS processes, it is sufficient to consider subordinators with $\EXP{L(t)}=t$ in which case form a one-parameter family. Accordingly, it is convenient to represent their \Levy\ density as follows (see Cont and Tankov\mycite{ContTankov2004} equation 4.19)
\begin{equation}
\ell_L(x) = \frac{1}{\Gamma(1-\alpha)}\left(\frac{1-\alpha}{\nu}\right)^{1-\alpha}\,\frac{e^{-(1-\alpha)/\nu}}{x^{1+\alpha}}\mathds{1}_{x\ge 0}\label{eq:cts:sub:density}
\end{equation}
where $\nu=\VAR{L(1)}$, $\beta=\frac{(1-\alpha)}{\nu}$ and $c=\frac{1}{\Gamma(1-\alpha)}\left(\frac{1-\alpha}{\nu}\right)^{1-\alpha}$. 
Furthermore, denoting $\varphi_Y(u, t)$ the \chf\ of $Y(t)$ and $\psi_Y(u) = t^{-1}\log(\varphi_Y(u, t))$ the characteristic exponent (\che) of $Y(\cdot)$, based on\refeqq{eq:chf:nvmm} it results
\begin{equation}
\psi_Y(u) = \frac{1 - \alpha}{\alpha\,\nu}\left[1 - \left(1 + \frac{\nu(u^2\sigma^2/2 - i\theta\,u)}{1-\alpha} \right)^{\alpha} \right].
\label{eq:che:ncts}
\end{equation}

\subsection{OU-TS processes}\label{sub:sec:ou:ts}
Consider an Ornstein-Uhlenbeck (OU) process $X(\cdot)$ solution of the stochastic differential equation
            \begin{eqnarray}\label{eq:genOU_sde}
              dX(t) &=&  -bX(t)dt + dL(t), \quad\qquad X(0)=X_0\quad \Pqo\qquad
              b>0
            \end{eqnarray}
namely,
\begin{eqnarray}
X(t) &=& X_0\,e^{-bt} + Z_L(t),	 \qquad\quad Z_L(t)=\int_0^te^{-b\,(t-s)}dL(s) \label{eq:sol:OU}
\end{eqnarray}
where $L(\cdot)$ is a \Levy\ process, named \emph{Background Driving \Levy\ Process} (BDLP) accordingly.
Following the convention in Barndorff-Nielsen and Shephard\mycite{BNSh01}, if $L(\cdot)$ is a TS subordinator $X(\cdot)$ is called OU-TS process.

There is a close relation between the concept of self-decomposability and the theory of \Levy-driven OU processes, indeed as observed in Barndorff-Nielsen et al.\mycite{BJS1998}, the solution process\refeqq{eq:sol:OU} is stationary if and only if its
\chf $\varphi_X(u,t)$ is constant in time and steadily coincides
with the \chf\ $\overline{\varphi}_X(u)$ of the \sd\ invariant
initial distribution that turns out to be decomposable according to
\begin{equation*}
\overline{\varphi}_X(u)=\overline{\varphi}_X(u\,e^{-b\,
t})\varphi_Z(u,t)
\end{equation*}
where now, at every given $t$, $\varphi_Z(u,t)=e^{\psi_Z(u,t)}$ denotes the \chf\ of the \rv\ $Z(t)$ in\refeqq{eq:sol:OU} and $\psi_Z(u,t)$ its \lch. We remark that the process $Z(\cdot)$ is not a \Levy\ process, but rather an additive process. 

We recall that a law with \chf\
$\eta(u)$ is said to be \sd\ (see Sato\mycite{Sato} or Cufaro
Petroni~\cite{Cufaro08}) when for every $0<a<1$ we can find another
law with \chf\ $\chi_a(u)$ such that
\begin{equation}\label{aremchf}
    \eta(u)=\eta(au)\chi_a(u).
\end{equation}
Of course, a \rv\ $X$ with \chf\ $\eta(u)$ is also
said to be \sd\ when its law is \sd, and looking at the definitions
this means that for every $0<a<1$ we can always find two
\emph{independent} \rv's -- a $Y$ with the same law of $X$, and a
$Z_a$ with \chf\ $\chi_a(u)$ -- such that in distribution
\begin{equation}\label{sdec-rv}
    X\eqd aY+Z_a
\end{equation}
Hereafter the \rv\ $Z_a$ will be called the \emph{\arem} of $X$ and
in general has an \id\ (see Sato\mycite{Sato}).
This last statement apparently means that the law of $Z_L(t)$ in the
solution\refeqq{eq:sol:OU} coincides with that of the \arem\ of the
\sd, stationary law $\overline{\varphi}_X$ provided that $a=e^{-b\,
t}$. It is easy indeed to see
from\refeqq{eq:sol:OU} that the \lch\ $\psi_X(u,t)$ of the time homogeneous
transition law with a degenerate initial condition $X(0)=x_0,\;\Pqo$
is
\begin{equation}
    \psi_X(u,t|x_0)=iux_0e^{-bt} + \psi_{Z_L}(u,t).
\label{eq:cumulant:function}
\end{equation}
and can also
be written in terms of the corresponding \che\ $\psi_L(u)$ in
the form
\begin{equation}
    \psi_X(u,t|x_0)=iux_0e^{-bt} + \psi_{Z_L}(u,t) = iux_0e^{-b\, t} + \int_0^t\psi_L\left(ue^{-b\,s}\right)ds.
\label{eq:lch:ou}
\end{equation}
\section{Symmetric OU-NTS processes and their exact simulation}\label{sec:ou:ncts}
In this section we focus on OU processes whose BDLP $Y(\cdot)$ is a symmetric NTS process (hereafter dubbed SNTS) with $\theta=0$ and with $L(\cdot)$ a TS subordinator. Although the law of the SNTS process at a fix time $t$ boils down into a three-parameters distribution, in the following we will keep the notation $\sncts(\sigma, \alpha, \beta, c\,t)$ with four parameters under the assumption that 
\begin{equation*}
\beta=\frac{1-\alpha}{\nu},\quad\quad c = \frac{1}{\Gamma(1-\alpha)}\left(\frac{1-\alpha}{\nu}\right)^{1-\alpha}=\frac{\beta^{1-\alpha}}{\Gamma(1-\alpha)}.
\end{equation*}
In virtue of\refeqq{eq:sol:OU} we write
\begin{eqnarray}
N(t) &=& N(0)\,e^{-bt} + Z_Y(t)\label{eq:sol:OUNCTS_N},\\
Z_Y(t)&=&\int_0^te^{-b\,(t-s)}dY(s) = \sigma\int_0^te^{-b\,(t-s)}dW(L(s)) \label{eq:sol:OUNCTS_Z}
\end{eqnarray}
where $N(\cdot)$ denotes a OU-SNTS process relatively to which the following proposition holds.
\begin{prop}\label{prop:rv:ou:sncts}
Denoting $a=e^{-b\,t}$ and $\omega=a^2$, the pathwise solution\refeqq{eq:sol:OUNCTS_N} with $N(0)=N_0, \Pqo$, is in distribution the sum of three independent \rv's
\begin{equation}
N(t) = a\,N_0 + Z_Y(t) \eqd a\,N_0 + N_1 + N_2
\label{eq:prop:ou:sncts:dec}
\end{equation}
where 
\begin{equation}
N_1 \eqd \sigma\sqrt{M_1}\,X_1
\label{eq:prop:ou:sncts:n1}
\end{equation}
with $X_1$, $M_1$ independent, $X_1\sim\norm(0, 1)$ and $M_1\sim\cts\!\left(\alpha,
\frac{\beta}{\omega}, c\,\frac{1 - \omega^\alpha}{2\,\alpha\, b}\right)$,
therefore $N_1$ is distributed according to a $\sncts\left(\sigma, \alpha,
\frac{\beta}{\omega}, c\,\frac{1 - \omega^\alpha}{2\,\alpha\, b}\right)$ law. Moreover, 
\begin{equation}
N_2 \eqd \sigma\sqrt{M_2}\,X_2
\label{eq:prop:ou:sncts:n2}
\end{equation}
with $X_2$, $M_2$ independent and independent of $X_1$ and $M_1$. In its turn,  
\begin{equation*}
M_2 =\sum_{k=0}^{P_\omega} J_k, \quad J_0=0, \Pqo
\end{equation*}
is a compound Poisson \rv\ where $P_\omega$ is an independent Poisson
\rv\ with parameter
\begin{equation}
 \Lambda_\omega =
 \frac{c\,\beta^{\alpha}\Gamma(1-\alpha)}{2\,b\,\alpha^2 \omega^{\alpha}}\,\left(1-\omega^{\alpha}+\omega^{\alpha}\log \omega^{\alpha}\right) \label{eq:ou:ncts:intensity}
\end{equation}
and $J_k, k>0$ are \iid\ \rv's with density
\begin{equation}
 f_J(x)=\frac{\alpha\, \omega^{\alpha}}{1-\omega^\alpha+\omega^{\alpha}\log \omega^{\alpha}} \int_1^{\frac{1}{\omega}}\frac{x^{-\alpha} \left(\beta\,
  v\right)^{1-\alpha}e^{-\beta v\, x}}{\Gamma(1-\alpha)}\, \frac{v^{\alpha}-1}{v}\, dv \label{eq:ouncts:jumps}
\end{equation}
namely, a mixture of a gamma law and a distribution with density
\begin{equation}
  f_V(v) = \frac{\alpha\, \omega^\alpha}{1-\omega^{\alpha}+\omega^{\alpha}\log
    \omega^{\alpha}}\, \frac{v^{\alpha}-1}{v}\qquad\quad 1\le v\le\,^1/_\omega.\label{eq:ou:sncts:mixture:rate}
\end{equation}
\end{prop}

\begin{proof}
Based on the definition of the \chf\ $\varphi_{Z_Y}(u, t)$ of $Z_Y(t)$ and on the fact that $L(\cdot)$ and $W(\cdot)$ are independent, it results
\begin{eqnarray}
\varphi_{Z_Y}(u, t) &=& \EXP{exp\left(i\,u\, \sigma^2 \int_0^te^{-\,b\,(t-s)} dW(L(s)) \right)} = \nonumber\\
&& \EXP{\EXP{exp\left(i\,u\, \sigma^2\int_0^te^{-\,b\,(t-s)} dW(L(s)) \right)}\Bigg|L(t)} = \nonumber\\
&& \EXP{exp\left(-\frac{\sigma^2\,u^2}{2}\int_0^te^{-2\,b\,(t-s)} dL(s) \right)} = \zeta\left(\frac{i\,\sigma^2\,u^2}{2}, t\right),\label{eq:zeta}
\end{eqnarray}
where $\zeta(u, t)$ represents the \chf\ of $\tilde{Z}_L(t) = \int_0^te^{-2\,b\,(t-s)} L(s)$, therefore $Z_Y(t)$ is distributed according to a normal variance-mean mixture distribution whose mixture law is that with \chf\ $\zeta(u, t)$.
On the other hand, from Proposition 5.1 in Cufaro Petroni and Sabino\mycite{cs20_3}, $\tilde{Z}_L(t)$ can be written as the sum of two independent \rv's 
\begin{equation*}
\tilde{Z}_L(t) \eqd M_1 + M_2,\quad\quad \zeta(u, t) = \zeta_1(u)\, \zeta_2(u)
\end{equation*}
where $\zeta_1(u)$ is the \chf\ of $M_1\sim\cts\!\left(\alpha,
\frac{\beta}{\omega}, c\,\frac{1 - \omega^\alpha}{2\,\alpha\, b}\right)$,
whereas $\zeta_2(u)$ is the \chf\ of $M_2$ in the hypothesis.
Consequently $\varphi_{Z_Y}(u, t)$ has the form
\begin{equation*}
\varphi_{Z_Y}(u, t) = \zeta_1\left(\frac{i\,\sigma^2\,u^2}{2}\right)\zeta_2\left(\frac{i\,\sigma^2\,u^2}{2}\right)
\end{equation*}
hence
\begin{equation*}
Z_Y(t) \eqd N_1 + N_2
\end{equation*}
where $N_1$ and $N_2$ are two independent \rv's each distributed according to a normal variance mixture law whose mixture laws are the distributions of $M_1$ and $M_2$, respectively and that concludes the proof. 
\end{proof}
\begin{remark}
Because of\refeqq{eq:zeta}, it also holds
\begin{equation}
Z_Y(t) \eqd X\,\sqrt{M_1 + M_2}
\label{eq:rem:ou:sncts}
\end{equation}
 where $X\sim\norm(0, 1)$ and it turns out to be a more convenient and a computationally efficient way to simulate the skeleton of $Y(\cdot)$.
\end{remark}
\begin{prop}\label{prop:chf:ou:sncts}
The \lch\ $\psi_{Z_Y}(u, t)$, $u\in\R$ with $0<\alpha<1$ can be represented as:
\begin{equation}
\psi_{Z_Y}(u, t) = -\frac{c\,\beta^{\alpha}\,\Gamma(1-\alpha)}{2\,\alpha\,b}\left[J\left(u, \sigma, \alpha, \beta, \frac{\beta}{\omega}\right) + \log \omega\right],\quad \omega=e^{-2\,b\,t}
\label{eq:lch:ou:cts}
\end{equation}
with
\begin{eqnarray}
J\left(u, \sigma, \alpha, \beta_1, \beta_2\right) &=& \int_{\beta_1}^{\beta_2} z^{-1-\alpha}\left(z + \frac{\sigma^2}{2}\,u\right)^{\alpha}dz=\nonumber\\
&=&-\frac{1}{\alpha}\left[\left(\frac{\sigma^2\,u^2}{2\,\beta_2}\right)^{\alpha}\, _2F_1\left(-\alpha, -\alpha, 1-\alpha, -\frac{2\,\beta_2}{\sigma^2\,u^2}\right) \right.- \nonumber\\
&&\left.\left(\frac{\sigma^2\,u^2}{2\,\beta_1}\right)^{\alpha}\,_2F_1\left(-\alpha, -\alpha, 1-\alpha, -\frac{2\,\beta_1}{\sigma^2\,u^2}\right)  \right] 
\label{eq:J_W}
\end{eqnarray}
\noindent where $_2F_1(a, b, c, x)$ is the hypergeometric function, $0<\alpha<1$, $\beta_1>0$ and $\beta_2>\beta_1$.
Finally, taking $\alpha=1/2$ the hypergeometric function in\refeqq{eq:J_W} can be expressed in terms of elementary functions as
\begin{equation*}
_2F_1\left(-\frac{1}{2}, -\frac{1}{2}, \frac{1}{2}, -x\right) =\sqrt{x + 1} - \sqrt{x}\sinh^{-1}\left(\sqrt{x}\right) 
\end{equation*}
\noindent where $\sinh^{-1}x = \log\left(\sqrt{x^2 + 1} + x\right)$.
\end{prop}
\begin{proof}
Because of Proposition 3.2 in Sabino\mycite{Sabino21a}, the \lch\ $\rho(u, t)$ of $\tilde{Z}_L(t)= \int_0^te^{-2\,b\,(t-s)} L(s)$ is
\begin{equation}
\rho(u, t) = -\frac{c\,\beta^{\alpha}\,\Gamma(1-\alpha)}{2\,\alpha\,b}\left[I\left(u, \alpha, \beta, \frac{\beta}{\omega}\right) + \log \omega\right],\quad \omega=e^{-2\,b\,t}
\label{eq:lch:zeta:tilede}
\end{equation}
with
\begin{eqnarray}
I\left(u, \alpha, \beta_1, \beta_2\right) &=& \int_{\beta_1}^{\beta_2} z^{-1-\alpha}(z - iu)^{\alpha}dz=\nonumber\\
&=&-\frac{1}{\alpha}\left[\left(\frac{u}{i\,\beta_2}\right)^{\alpha}\, _2F_1\left(-\alpha, -\alpha, 1-\alpha, -\frac{i\,\beta_2}{u}\right) \right.- \nonumber\\
&=&\left.\left(\frac{u}{i\,\beta_1}\right)^{\alpha}\,_2F_1\left(-\alpha, -\alpha, 1-\alpha, -\frac{i\,\beta_1}{u}\right)  \right]. 
\end{eqnarray}
On the other hand, we know from\refeqq{eq:zeta} that $\varphi_{Z_Y}(u, t)=\zeta\left(\frac{i\,\sigma^2\,u^2}{2}, t\right)$ therefore, simply adapting this form to the \lch\ $\rho(u, t)$ concludes the proof.  
\end{proof}
\begin{remark}
Using the transformation 9.131.1 in Gradshteyn and Rizhik\mycite{gradshteyn2007} we have
\begin{equation*}
_2F_1(-\alpha, -\alpha, 1-\alpha, x) = (1-x)^{\alpha + 1}\,_2F_1(1, 1, 1-\alpha, x)
\end{equation*} 
and accordingly\refeqq{eq:J_W} becomes
\begin{eqnarray}
I\left(u, \alpha, \beta_1, \beta_2\right) 
&=&-\frac{i}{\alpha\,u}\left[\beta_2 ^{-\alpha}\left(\beta_2 - i\,u\right)^{\alpha+1} \, _2F_1\left(1, 1, 1-\alpha, -\frac{i\,\beta_2}{u}\right) \right.- \nonumber\\
&&\left.\beta_1 ^{-\alpha}\left(\beta_1 - i\,u\right)^{\alpha+1}\, _2F_1\left(1, 1, 1-\alpha, -\frac{i\,\beta_1}{u}\right) \right].
\label{eq:J_W_2}
\end{eqnarray}
We make use of this different representation in the following corollary.
\end{remark}
\begin{cor}\label{corr:mgf:ou:sncts}
The \cgf\ $m_{Z_Y}(s,t)= \ln\EXP{e^{s\,Z_Y(t)}}$ of $Z_Y(t)$ with $0<\alpha<1$ exists for $-\frac{\sqrt{2\,\beta}}{\sigma} < s < \frac{\sqrt{2\,\beta}}{\sigma}$ and is:
\begin{equation}\label{eq:mgf:ou:sncts}
 m_{Z_Y}(s,t) = -\frac{c\,\beta^{\alpha}\,\Gamma(1-\alpha)}{2\,\alpha\,b}\left[\tilde{J}\left(s, \sigma, \alpha, \beta, \frac{\beta}{\omega}\right) + \log \omega\right] 
\end{equation}
where
\begin{eqnarray}
\tilde{J}\left(s, \sigma, \alpha, \beta_1, \beta_2\right) &=& \int_{\beta_1}^{\beta_2} z^{-1-\alpha}\left(z - \frac{\sigma^2\,s^2}{2}\right)^{\alpha}dz=\nonumber\\
&&
\frac{2}{\alpha\,\sigma^2\,s^2}\left[\beta_2 ^{-\alpha}\left(\beta_2 - \frac{\sigma^2\,s^2}{2}\right)^{\alpha+1} \, _2F_1\left(1, 1, 1-\alpha, \frac{2\,\beta_2}{\sigma^2\,s^2}\right) \right.- \nonumber\\
&&\left.\beta_1 ^{-\alpha}\left(\beta_1 - \frac{\sigma^2\,s^2}{2}\right)^{\alpha+1}\, _2F_1\left(1, 1, 1-\alpha, -\frac{2\,\beta_1}{\sigma^2s^2}\right) \right]\label{eq:I_tilde}
\end{eqnarray}
\end{cor}


\subsection{Simulation Algorithms}\label{sub:sec:ou:sncts:alg}
Proposition\myref{prop:rv:ou:sncts} provides the theoretical basis for the sequential generation of the skeleton of $N(\cdot)$ on a time grid $t_0, t_1, \dots t_M$, that, assuming at each step $a_m=e^{-b(t_{m}-t_{m-1})}, m=0, \dots, M$, can be accomplished by the following recursion with initial condition $N(t_0)=n_0$:
\begin{equation}
N(t_{m}) = a_m N(t_{m-1}) + Z_Y(t_{m-1}), \quad\qquad m=1,\dots, M,
\label{eq:ou:ncts:recursion}
\end{equation}
the implementation steps are also detailed in Algorithm\myref{alg:sim:ou:sncts}.
\begin{algorithm}
\caption{ }\label{alg:sim:ou:sncts}
\begin{algorithmic}[1]
\State  $N(t_0) \gets n_0$.
		\For{ $m=1, \dots, M$}
		\State $a_m=e^{-b(t_{m}-t_{m-1})}$
		\State $\omega_m=a_m^2 = e^{-2\,b(t_{m}-t_{m-1})}$
		\State $x\gets X\sim\norm(0, 1)$. 
		\State $m_1\gets M_1\sim\cts\!\left(\alpha,
\frac{\beta}{\omega_m}, c\,\frac{1 - \omega_m{\alpha}}{2\,\alpha\, b}\right)$
		\State $p\gets P_{\omega_m}\sim\poiss(\Lambda_{\omega_m})$,
        \Comment{Generate an independent Poisson \rv\ with $\Lambda_{\omega_m}$ in\refeqq{eq:ou:ncts:intensity}}
        \State $v_i\gets V_i, \quad i=1, \dots, p$\Comment{Generate independent \rv's with \pdf\refeqq{eq:ou:sncts:mixture:rate}}
        \State $\tilde{\beta}_i \gets \beta\, v_i, \quad i=1, \dots, p$
        \State $j_i\gets J_i\sim\gam(1-\alpha,\, \tilde{\beta}_i), \quad i=1, \dots, n$
        \Comment{Generate $p$ independent gamma \rv's all with the same scale $1-\alpha$ and random rates}
		\State $m_2\gets \sum_{i=1}^pj_i$ 
		\State $z_m\gets x\,\sqrt{m_1 + m_2}$
		\State $N(t_m)\gets a_i\,N(t_{m-1}) + z_m$.
		\EndFor
		\end{algorithmic}
\end{algorithm}

The generation of the rv's in Algorithm\myref{alg:sim:ou:sncts} is standard except that for the TS distributed \rv\ $M_1$ and for $V$. 
To this end, the sampling from a TS law has been widely studied by several
authors (see for instance Devroye\mycite{Dev2009} and
Hofert\mycite{Hofert2012}), of course, taking $\alpha=0.5$ the TS law coincides with an IG law hence, one can rely on the many-two-one transformation method of Michael et
al.\mycite{MSH76}. With regards to the simulation of $V$, an efficient algorithm is detailed in Cufaro Petroni and Sabino\mycite{cs20_3} which is based on the decomposition-rejection method illustrated in Devroye\mycite{Dev86} page 67.
Furthermore, we remark that using a Taylor expansion for $\Lambda_{\omega}$ we
find
\begin{equation*}
    \Lambda_\omega=\frac{c\Gamma(1-\alpha)b\,\beta^{\alpha}}{4}\,
\Delta\, t^2+o\big(\Delta\,t^2\big),
\end{equation*}
and therefore the compound Poisson $M_2$ has a relevant impact only
when $\Delta t$ is not too small. This observation gives then some hints on possible approximations. We will compare  the performance of our approach to two alternatives: the first simply ignores $M_2$ (Approximation 1), the second, in the same vein of Benth
et al.\mycite{BDPL18}, relies on the Euler scheme applied to the stochastic differential equation of an OU-SNTS process, namely it 
approximates the law of $Z_Y(\Delta t)$ with that of
$L_Y(\Delta t)\sim\sncts\left(\sigma, \alpha,\beta, c\,\Delta t\right)$ (Approximation 2) hence, in both cases $Z_Y(\Delta t)$ is approximated by a SNTS distributed \rv.

\subsection{Numerical Experiments}\label{sub:sec:ou:sncts:num}
In this section, we analyze the performance and the
effectiveness of the algorithms for the simulation of OU-SNTS processes. All the numerical studies in this paper
have been conducted using \emph{Python} with a $64$-bit Intel Core
i5-6300U CPU, 8GB. The performance of the algorithms is measured in
terms of the percentage error relatively to the second and fourth cumulants
denoted \emph{err} \% and defined as
\begin{equation*}
   \text{err \%} = \frac{\text{true value} - \text{estimated
value}}{\text{true value}}.
\end{equation*}
As a consequence of Equation\refeqq{eq:lch:ou} one can calculate the cumulants 
$c_{N,k}(n_0,t),\;k= 1,2,\ldots$ of $N(t)$ for $N_0=n_0$ from the
cumulants $c_{L_Y,k}$ of the symmetric NTS law according to
\begin{eqnarray}
c_{N,1}(n_0,t) &=&\EXP{N(t)|N_0=n_0}\;=\; n_0e^{-b\, t} + \frac{c_{L_Y,1}}{b}\left(1-e^{-b\, t}\right),\qquad k=1\label{eq:cumulants:ou1}\\
c_{N,k}(n_0,t) &=& \frac{c_{L_Y,k}}{k\,b}\left(1-e^{-k\,b\,
t}\right) \qquad\qquad\qquad\qquad\qquad
k=2,3,\ldots\label{eq:cumulants:ou2}
\end{eqnarray}
On the other hand, $c_{L_Y,k}=0$ when $k$ is an odd integer and setting $k=2\,n, n>0$ $c_{L_Y,2\,n}$ can be computed from the \Levy\ density as follows 
\begin{eqnarray}\label{eq:cts:cumulants}
c_{L_Y,2\,n} &=&\int_{-\infty}^{+\infty}x^{2\,n}\ell_{L_Y}(x)\,dx=\nonumber\\
&&C(\alpha, \nu, \sigma)\int_{-\infty}^{+\infty}x^{2\,n}|x|^{-\alpha -1/2} K_{\alpha+1/2}\left(A(\alpha, \nu, \sigma)\,|x|\right)\,dx=\nonumber\\
&& C(\alpha, \nu, \sigma) 2^{2\,n-\alpha-1/2}\,A(\alpha, \nu, \sigma)^{-2\,n-\alpha-1/2}\,\Gamma\left(\frac{2\,n+1}{2}\right)\,\Gamma\left(n-\alpha\right), 
\end{eqnarray}
where $K_{\eta}(x)$ represents the modified Bessel function of the second type, 
\begin{eqnarray*}
C(\alpha, \nu, \sigma) &=& \frac{2^{\alpha/2+5/4}}{\sqrt{2\,\pi}\,\Gamma(1-\alpha)}\sigma^{\alpha-1/2}\beta^{5/4-\alpha/2},\\
A(\alpha, \nu, \sigma) &=& \frac{\sqrt{2\,\beta}}{\sigma}
\end{eqnarray*}
and in the last step we have used 6.561.16 at page 676 in Gradshteyn and Rizhik\mycite{gradshteyn2007}. Of course, with $\alpha=1/2$ we retrieve the cumulants of a symmetric NIG law shown in Table 4.5 in Cont and Tankov\mycite{ContTankov2004}.

\begin{figure}
        \begin{subfigure}[c]{.5\textwidth}{
                \includegraphics[width=70mm]{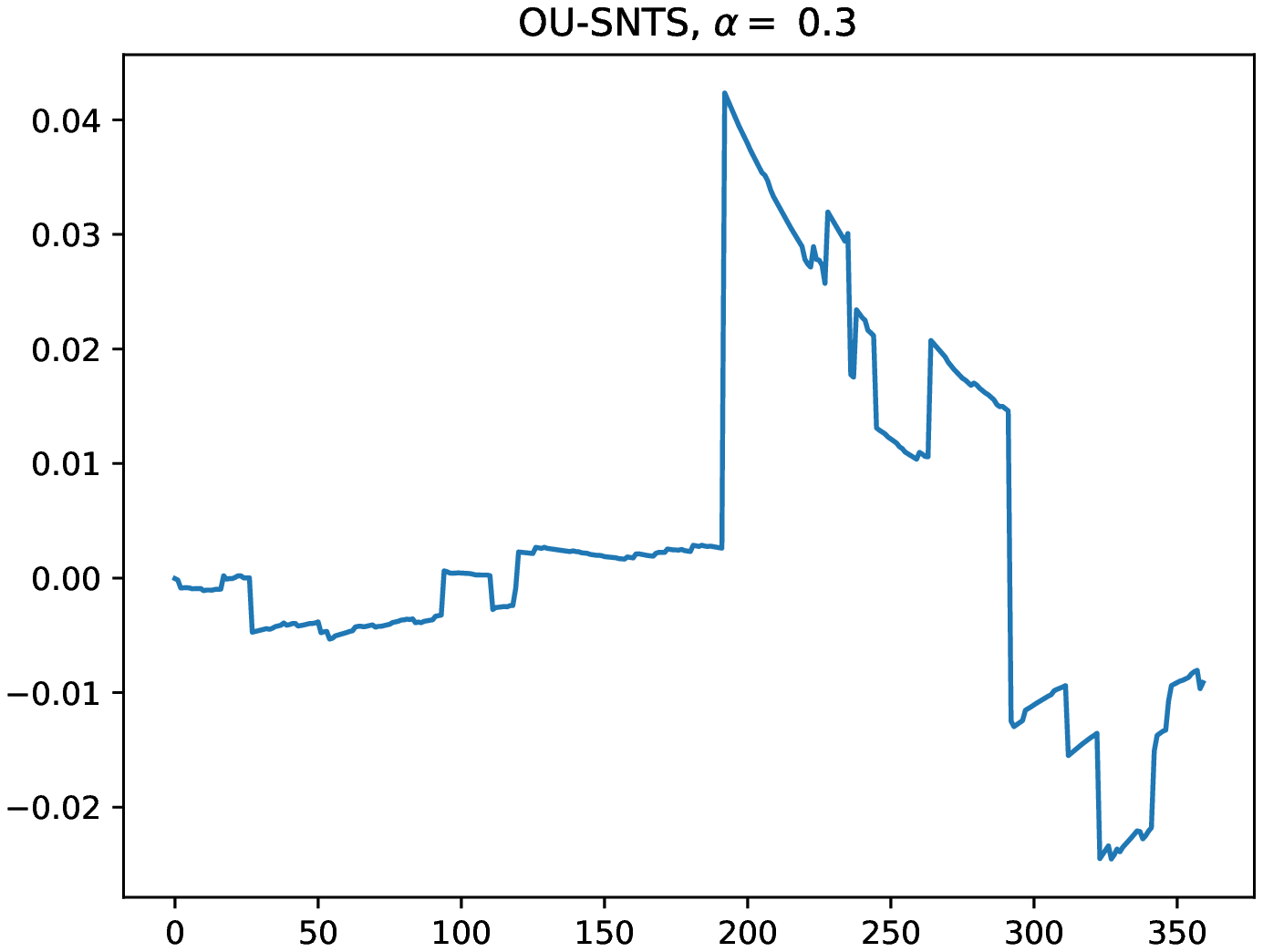}
                }
        \end{subfigure}
        \begin{subfigure}[c]{.5\textwidth}{
                \includegraphics[width=70mm]{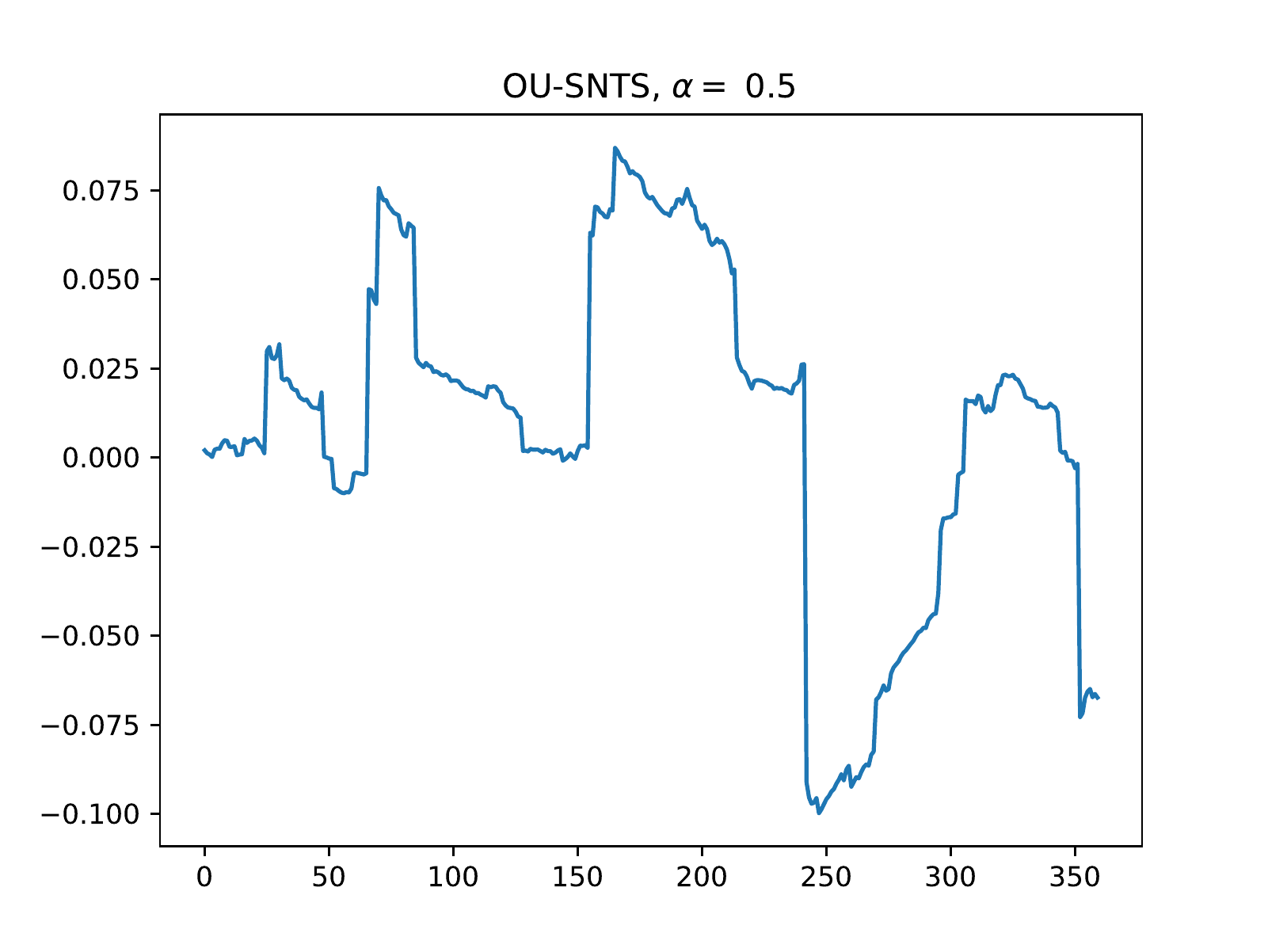}
                }
        \end{subfigure}
			\\
        \begin{subfigure}[c]{.5\textwidth}{
                \includegraphics[width=70mm]{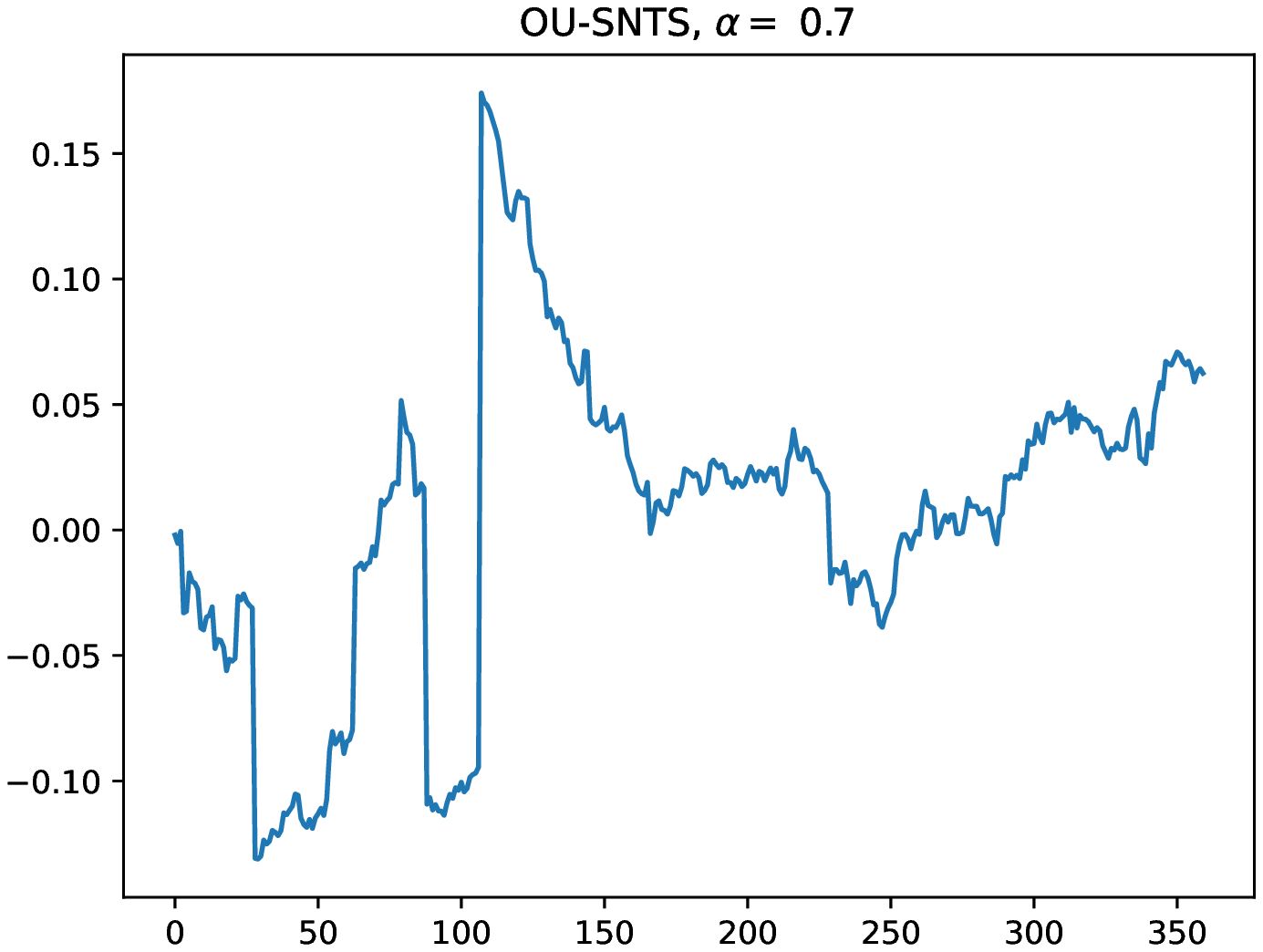}
                }
        \end{subfigure}
        \begin{subfigure}[c]{.5\textwidth}{
                \includegraphics[width=70mm]{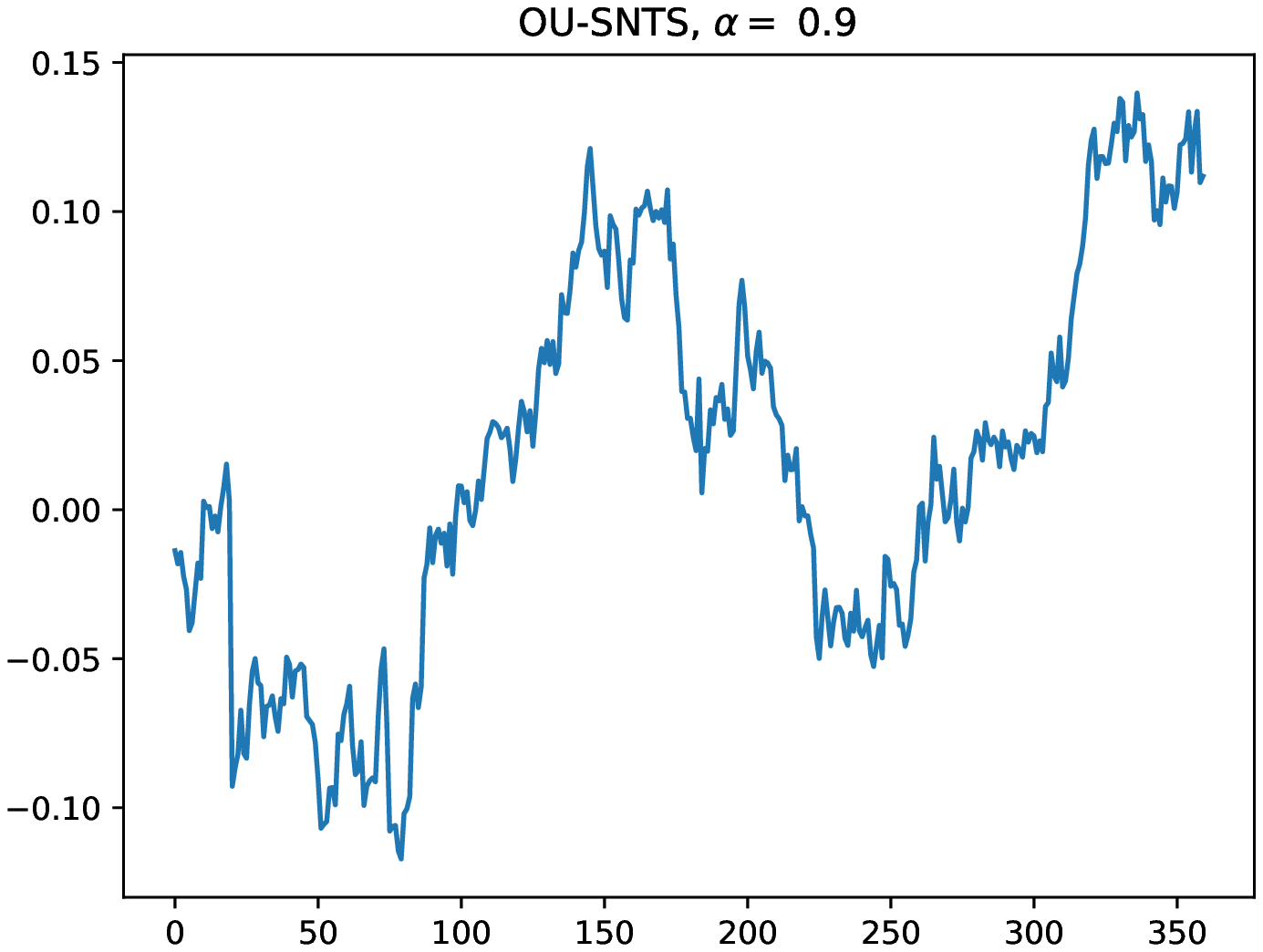}
                }
        \end{subfigure}
\caption{Sample trajectories of a OU-SNTS process with $\left(b, \sigma, \nu
\right) = \left(5, 0.3, 2.5\right)$.}\label{fig:ou:sncts}				
\end{figure}

Tables\myref{tab:exact:1:360} and\myref{tab:exact:1:12} show the comparison between the true values of $c_{N,2}(n_0,\Delta t)$ and $c_{N,4}(n_0,\Delta t)$ and their corresponding MC estimates. We have considered $n_0=0$, $\left(b, \sigma, \nu\right) = \left(5, 0.3, 2.5\right)$ with $\Delta t=1/365$ and $\Delta t=1/12$, respectively, and let $\alpha$ vary; similar conclusions have been derived with different parameter sets but are not reported for brevity. Moreover, Figure\myref{fig:ou:sncts} displays some path trajectories changing the parameter $\alpha$.

Irrespective to the values of $\Delta t$ and $\alpha$, one can observe that the convergence is relatively slow and a large number of simulation is required to get reliable estimates. Indeed, with the exact  scheme $10^5$ simulations are needed to achieve roughly five percent of percentage error both for the second and forth cumulant. 

Besides this fact, if $\Delta t$ equals $1/365$ (one day in financial applications) the approximated schemes and the exact scheme return similar estimates with little bias, although Approximation 1 seems better than Approximation 2. It is well known that the quality of the MC estimation deteriorates with the increase of the order of the cumulants, indeed in our numerical example we need $10^6$ scenarios to get a small percentage error.

In contrast, when $\Delta t=1/12$ - e.g. one month - one can observe in Table\myref{tab:exact:1:12} that the estimates of the approximated schemes are totally biased and no matter how high is the number of simulations, the MC estimates do not converge to the true values. On the other hand, the exact scheme performs well and exhibits a similar behavior to the one with a small $\Delta t$. 

These observations give a few hints on how to conceive the parameter calibration for a OU-SNTS process and provide guidance for the path generation. For instance, one can take a time series with relatively small frequency, e.g. daily, and estimate the parameters assuming that the transition density is that of SNTS according to Approximation 1. In particular for $\alpha=1/2$, the law is a symmetric NIG and some software packages based on the maximum-likelihood method are available to accomplish the parameter estimation\footnote{See for instance, the package \emph{GeneralizedHyperbolic} in R.}. On the other hand, if the skeleton of the process is to be generated on a not so fine time grid, the exact scheme is the only valid alternative. More important, the pricing of forward start options, for instance those contracts whose first settlement date is in one month, do not require the simulation of the process before the first settlement date, leading to a remarkable computational advantage even though the exact scheme relies on an additional step compared to the two approximated procedure. We will comment more of these facts in the following sections dedicated to the financial applications.

\begin{table}[ht!]
    \centering\scriptsize
        \resizebox{\textwidth}{!}{		\begin{tabular}{*{13}{|c|rr|rr|rr|rr|rr|rr}}
				\hline
				\multicolumn{13}{|c|}{$\alpha=0.1$}\\
				\hline
&	\multicolumn{6}{c}{$c_{N,2}(n_0=0, \Delta\,t)=2.466\times 10^{-4}$}			&			\multicolumn{6}{|c|}{$c_{N,2}c_{N,n}(n_0=0, \Delta\,t)=1.641\times 10^{-4}$}\\
\hline
& \multicolumn{2}{c}{\text{Exact}}	&	\multicolumn{2}{c|}{\text{Approximation 1}}	&	\multicolumn{2}{c}{\text{Approximation 2}}	&	\multicolumn{2}{|c|}{\text{Exact}} & \multicolumn{2}{c|}{\text{Approximation 1}} &	\multicolumn{2}{c|}{\text{Approximation 2}} \\ \hline	
$N$	& \text{MC} &	\text{err} \%	& \text{MC} &	\text{err} \%	 & \text{MC}	& \text{err} \%	& \text{MC} &	\text{err} \%	& \text{MC}	& \text{err} \%	& \text{MC}	& \text{err} \%\\
\hline
$1000$ & $2.972$ & $-20.5\%$ & $3.524$ & $-42.9\%$ & $1.344$ & $45.5\%$ & $0.785$ & $52.2\%$ & $0.914$ & $44.3\%$ & $0.935$ & $43.0\%$\\
$10000$ & $2.043$ & $17.2\%$ & $3.268$ & $-32.5\%$ & $3.313$ & $-34.3\%$ & $0.796$ & $51.5\%$ & $0.888$ & $45.9\%$ & $0.918$ & $44.1\%$\\
$20000$ & $2.233$ & $9.4\%$ & $2.048$ & $17.0\%$ & $1.813$ & $26.5\%$ & $1.284$ & $21.8\%$ & $1.908$ & $-16.3\%$ & $2.076$ & $-26.5\%$\\
$50000$ & $2.577$ & $-4.5\%$ & $2.109$ & $14.5\%$ & $2.788$ & $-13.1\%$ & $1.384$ & $15.7\%$ & $1.464$ & $10.8\%$ & $2.032$ & $-23.8\%$\\
$100000$ & $2.503$ & $-1.5\%$ & $2.320$ & $5.9\%$ & $2.580$ & $-4.6\%$ & $1.464$ & $10.8\%$ & $1.494$ & $9.0\%$ & $1.821$ & $-11.0\%$\\
$1000000$ & $2.483$ & $-0.7\%$ & $2.515$ & $-2.0\%$ & $2.418$ & $2.0\%$ & $1.684$ & $-2.6\%$ & $1.550$ & $5.5\%$ & $1.532$ & $6.6\%$\\
\hline
				\multicolumn{13}{|c|}{$\alpha=0.3$}\\
				\hline
&	\multicolumn{6}{c}{$c_{N,2}(n_0=0, \Delta\,t)=2.466\times 10^{-4}$}			&			\multicolumn{6}{|c|}{$c_{N,2}c_{N,n}(n_0=0, \Delta\,t)=1.641\times 10^{-4}$}\\
\hline
& \multicolumn{2}{c}{\text{Exact}}	&	\multicolumn{2}{c|}{\text{Approximation 1}}	&	\multicolumn{2}{c}{\text{Approximation 2}}	&	\multicolumn{2}{|c|}{\text{Exact}} & \multicolumn{2}{c|}{\text{Approximation 1}} &	\multicolumn{2}{c|}{\text{Approximation 2}} \\ \hline	
$N$	& \text{MC} &	\text{err} \%	& \text{MC} &	\text{err} \%	 & \text{MC}	& \text{err} \%	& \text{MC} &	\text{err} \%	& \text{MC}	& \text{err} \%	& \text{MC}	& \text{err} \%\\		
\hline						
$1000$ & $0.196$ & $92.1\%$ & $1.009$ & $59.1\%$ & $1.249$ & $49.4\%$ & $0.665$ & $59.4\%$ & $0.187$ & $88.6\%$ & $0.890$ & $45.7\%$\\
$10000$ & $2.915$ & $-18.2\%$ & $3.299$ & $-33.8\%$ & $3.299$ & $-33.8\%$ & $0.957$ & $41.7\%$ & $2.261$ & $-37.8\%$ & $0.957$ & $41.7\%$\\
$20000$ & $2.717$ & $-10.2\%$ & $2.006$ & $18.6\%$ & $3.120$ & $-26.5\%$ & $1.969$ & $-20.0\%$ & $1.980$ & $-20.7\%$ & $2.030$ & $-23.7\%$\\
$50000$ & $2.293$ & $7.0\%$ & $2.171$ & $12.0\%$ & $2.096$ & $15.0\%$ & $1.380$ & $15.9\%$ & $1.361$ & $17.1\%$ & $1.892$ & $-15.3\%$\\
$100000$ & $2.341$ & $5.1\%$ & $2.619$ & $-6.2\%$ & $2.239$ & $9.2\%$ & $1.519$ & $7.4\%$ & $1.871$ & $-14.0\%$ & $1.815$ & $-10.6\%$\\
$1000000$ & $2.424$ & $1.7\%$ & $2.513$ & $-1.9\%$ & $2.384$ & $3.3\%$ & $1.603$ & $2.3\%$ & $1.607$ & $2.1\%$ & $1.515$ & $7.7\%$\\
\hline
				\multicolumn{13}{|c|}{$\alpha=0.5$}\\
				\hline
&	\multicolumn{6}{c}{$c_{N,2}(n_0=0, \Delta\,t)=2.466\times 10^{-4}$}			&			\multicolumn{6}{|c|}{$c_{N,2}c_{N,n}(n_0=0, \Delta\,t)=1.641\times 10^{-4}$}\\
\hline
& \multicolumn{2}{c}{\text{Exact}}	&	\multicolumn{2}{c|}{\text{Approximation 1}}	&	\multicolumn{2}{c}{\text{Approximation 2}}	&	\multicolumn{2}{|c|}{\text{Exact}} & \multicolumn{2}{c|}{\text{Approximation 1}} &	\multicolumn{2}{c|}{\text{Approximation 2}} \\ \hline	
$N$	& \text{MC} &	\text{err} \%	& \text{MC} &	\text{err} \%	 & \text{MC}	& \text{err} \%	& \text{MC} &	\text{err} \%	& \text{MC}	& \text{err} \%	& \text{MC}	& \text{err} \%\\		
\hline	
$1000$ & $0.514$ & $79.1\%$ & $4.215$ & $-70.9\%$ & $1.437$ & $41.7\%$ & $0.617$ & $62.4\%$ & $0.110$ & $93.3\%$ & $2.991$ & $-82.3\%$\\
$10000$ & $1.552$ & $37.1\%$ & $1.275$ & $48.3\%$ & $3.332$ & $-35.1\%$ & $2.265$ & $-38.0\%$ & $2.381$ & $-45.1\%$ & $2.389$ & $-45.6\%$\\
$20000$ & $2.185$ & $11.4\%$ & $3.117$ & $-26.4\%$ & $3.032$ & $-22.9\%$ & $1.265$ & $22.9\%$ & $1.216$ & $25.9\%$ & $2.030$ & $-23.7\%$\\
$50000$ & $2.293$ & $7.0\%$ & $2.171$ & $12.0\%$ & $2.171$ & $12.0\%$ & $1.915$ & $-16.7\%$ & $1.892$ & $-15.3\%$ & $1.890$ & $-15.2\%$\\
$100000$ & $2.341$ & $5.1\%$ & $2.247$ & $8.9\%$ & $2.247$ & $8.9\%$ & $1.761$ & $-7.3\%$ & $1.471$ & $10.3\%$ & $1.510$ & $8.0\%$\\
$1000000$ & $2.513$ & $-1.9\%$ & $2.531$ & $-2.6\%$ & $2.350$ & $4.7\%$ & $1.698$ & $-3.5\%$ & $1.699$ & $-3.5\%$ & $1.516$ & $7.6\%$\\
\hline
				\multicolumn{13}{|c|}{$\alpha=0.3$}\\
				\hline
&	\multicolumn{6}{c}{$c_{N,2}(n_0=0, \Delta\,t)=2.466\times 10^{-4}$}			&			\multicolumn{6}{|c|}{$c_{N,2}c_{N,n}(n_0=0, \Delta\,t)=1.641\times 10^{-4}$}\\
\hline
& \multicolumn{2}{c}{\text{Exact}}	&	\multicolumn{2}{c|}{\text{Approximation 1}}	&	\multicolumn{2}{c}{\text{Approximation 2}}	&	\multicolumn{2}{|c|}{\text{Exact}} & \multicolumn{2}{c|}{\text{Approximation 1}} &	\multicolumn{2}{c|}{\text{Approximation 2}} \\ \hline	
$N$	& \text{MC} &	\text{err} \%	& \text{MC} &	\text{err} \%	 & \text{MC}	& \text{err} \%	& \text{MC} &	\text{err} \%	& \text{MC}	& \text{err} \%	& \text{MC}	& \text{err} \%\\		
\hline	
$1000$ & $0.923$ & $62.6\%$ & $1.264$ & $48.7\%$ & $1.276$ & $48.3\%$ & $0.010$ & $99.4\%$ & $0.532$ & $67.6\%$ & $0.377$ & $77.0\%$\\
$10000$ & $1.557$ & $36.9\%$ & $1.910$ & $22.5\%$ & $1.792$ & $27.3\%$ & $0.816$ & $50.3\%$ & $2.389$ & $-45.6\%$ & $2.168$ & $-32.1\%$\\
$20000$ & $2.185$ & $11.4\%$ & $1.881$ & $23.7\%$ & $1.850$ & $25.0\%$ & $1.620$ & $1.3\%$ & $1.222$ & $25.5\%$ & $1.317$ & $19.8\%$\\
$50000$ & $2.342$ & $5.0\%$ & $2.006$ & $18.6\%$ & $2.085$ & $15.5\%$ & $1.333$ & $18.8\%$ & $1.322$ & $19.4\%$ & $1.377$ & $16.1\%$\\
$100000$ & $2.559$ & $-3.8\%$ & $2.227$ & $9.7\%$ & $2.714$ & $-10.1\%$ & $1.560$ & $5.0\%$ & $1.652$ & $-0.7\%$ & $1.749$ & $-6.6\%$\\
$1000000$ & $2.440$ & $1.0\%$ & $2.407$ & $2.4\%$ & $2.375$ & $3.7\%$ & $1.639$ & $0.1\%$ & $1.654$ & $-0.8\%$ & $1.627$ & $0.9\%$\\
\hline
				\multicolumn{13}{|c|}{$\alpha=0.3$}\\
				\hline
&	\multicolumn{6}{c}{$c_{N,2}(n_0=0, \Delta\,t)=2.466\times 10^{-4}$}			&			\multicolumn{6}{|c|}{$c_{N,2}c_{N,n}(n_0=0, \Delta\,t)=1.641\times 10^{-4}$}\\
\hline
& \multicolumn{2}{c}{\text{Exact}}	&	\multicolumn{2}{c|}{\text{Approximation 1}}	&	\multicolumn{2}{c}{\text{Approximation 2}}	&	\multicolumn{2}{|c|}{\text{Exact}} & \multicolumn{2}{c|}{\text{Approximation 1}} &	\multicolumn{2}{c|}{\text{Approximation 2}} \\ \hline	
$N$	& \text{MC} &	\text{err} \%	& \text{MC} &	\text{err} \%	 & \text{MC}	& \text{err} \%	& \text{MC} &	\text{err} \%	& \text{MC}	& \text{err} \%	& \text{MC}	& \text{err} \%\\		
\hline	
$1000$ & $1.276$ & $48.3\%$ & $1.412$ & $42.7\%$ & $1.451$ & $41.2\%$ & $0.902$ & $45.1\%$ & $0.909$ & $44.6\%$ & $0.190$ & $88.4\%$\\
$10000$ & $2.006$ & $18.6\%$ & $2.170$ & $12.0\%$ & $1.805$ & $26.8\%$ & $2.190$ & $-33.5\%$ & $0.506$ & $69.2\%$ & $2.381$ & $-45.1\%$\\
$20000$ & $2.233$ & $9.4\%$ & $2.089$ & $15.3\%$ & $1.905$ & $22.8\%$ & $2.061$ & $-25.6\%$ & $2.244$ & $-36.8\%$ & $1.835$ & $-11.8\%$\\
$50000$ & $2.288$ & $7.2\%$ & $2.694$ & $-9.2\%$ & $2.066$ & $16.2\%$ & $1.835$ & $-11.8\%$ & $1.844$ & $-12.4\%$ & $1.279$ & $22.1\%$\\
$100000$ & $2.350$ & $4.7\%$ & $2.398$ & $2.8\%$ & $2.216$ & $10.1\%$ & $1.730$ & $-5.4\%$ & $1.744$ & $-6.3\%$ & $1.491$ & $9.2\%$\\
$1000000$ & $2.502$ & $-1.5\%$ & $2.519$ & $-2.1\%$ & $2.566$ & $-4.0\%$ & $1.699$ & $-3.5\%$ & $1.652$ & $-0.6\%$ & $1.601$ & $2.4\%$\\
\hline
		\end{tabular}
		}
	\caption{MC-estimated cumulants multiplied by $10000$ and comparison to their true values. $\left(b, \sigma, \nu
\right) = \left(5, 0.3, 2.5\right)$, $\Delta t=1/365$.}
	\label{tab:exact:1:360}
\end{table}
\clearpage
\begin{table}[ht]
    \centering\scriptsize
        \resizebox{\textwidth}{!}{		\begin{tabular}[ht]{*{13}{|c|rr|rr|rr|rr|rr|rr}}
				\hline
				\multicolumn{13}{|c|}{$\alpha=0.1$}\\
				\hline
&	\multicolumn{6}{c}{$c_{N,2}(n_0=0, \Delta\,t)=5.089\times 10^{-3}$}			&			\multicolumn{6}{|c|}{$c_{N,2}c_{N,n}(n_0=0, \Delta\,t)=2.461\times 10^{-3}$}\\
\hline
& \multicolumn{2}{c}{\text{Exact}}	&	\multicolumn{2}{c|}{\text{Approximation 1}}	&	\multicolumn{2}{c}{\text{Approximation 2}}	&	\multicolumn{2}{|c|}{\text{Exact}} & \multicolumn{2}{c|}{\text{Approximation 1}} &	\multicolumn{2}{c|}{\text{Approximation 2}} \\ \hline	
$N$	& \text{MC} &	\text{err} \%	& \text{MC} &	\text{err} \%	 & \text{MC}	& \text{err} \%	& \text{MC} &	\text{err} \%	& \text{MC}	& \text{err} \%	& \text{MC}	& \text{err} \%\\
\hline
$1000$ & $6.154$ & $-20.9\%$ & $2.350$ & $53.8\%$ & $2.702$ & $46.9\%$ & $2.951$ & $-19.9\%$ & $0.216$ & $91.2\%$ & $0.370$ & $85.0\%$\\
$10000$ & $5.156$ & $-1.3\%$ & $3.659$ & $28.1\%$ & $2.780$ & $45.4\%$ & $2.832$ & $-15.1\%$ & $1.018$ & $58.7\%$ & $0.501$ & $79.6\%$\\
$20000$ & $5.067$ & $0.4\%$ & $3.491$ & $31.4\%$ & $2.845$ & $44.1\%$ & $2.743$ & $-11.5\%$ & $0.939$ & $61.8\%$ & $0.564$ & $77.1\%$\\
$50000$ & $5.073$ & $0.3\%$ & $3.464$ & $31.9\%$ & $3.158$ & $37.9\%$ & $2.616$ & $-6.3\%$ & $0.937$ & $61.9\%$ & $0.842$ & $65.8\%$\\
$100000$ & $5.080$ & $0.2\%$ & $3.544$ & $30.4\%$ & $3.179$ & $37.5\%$ & $2.520$ & $-2.4\%$ & $1.026$ & $58.3\%$ & $0.873$ & $64.5\%$\\
$1000000$ & $5.083$ & $0.1\%$ & $3.419$ & $32.8\%$ & $3.266$ & $35.8\%$ & $2.426$ & $1.4\%$ & $0.962$ & $60.9\%$ & $1.010$ & $59.0\%$\\
\hline
				\multicolumn{13}{|c|}{$\alpha=0.3$}\\
				\hline
&	\multicolumn{6}{c}{$c_{N,2}(n_0=0, \Delta\,t)=5.089\times 10^{-3}$}			&			\multicolumn{6}{|c|}{$c_{N,2}c_{N,n}(n_0=0, \Delta\,t)=2.461\times 10^{-3}$}\\
\hline
& \multicolumn{2}{c}{\text{Exact}}	&	\multicolumn{2}{c|}{\text{Approximation 1}}	&	\multicolumn{2}{c}{\text{Approximation 2}}	&	\multicolumn{2}{|c|}{\text{Exact}} & \multicolumn{2}{c|}{\text{Approximation 1}} &	\multicolumn{2}{c|}{\text{Approximation 2}} \\ \hline	
$N$	& \text{MC} &	\text{err} \%	& \text{MC} &	\text{err} \%	 & \text{MC}	& \text{err} \%	& \text{MC} &	\text{err} \%	& \text{MC}	& \text{err} \%	& \text{MC}	& \text{err} \%\\		
\hline						
$1000$ & $5.723$ & $-12.5\%$ & $4.456$ & $12.4\%$ & $4.472$ & $12.1\%$ & $2.183$ & $11.3\%$ & $1.240$ & $49.6\%$ & $1.874$ & $23.9\%$\\
$10000$ & $5.827$ & $-14.5\%$ & $3.703$ & $27.2\%$ & $3.169$ & $37.7\%$ & $2.742$ & $-11.4\%$ & $0.826$ & $66.4\%$ & $0.659$ & $73.2\%$\\
$20000$ & $5.480$ & $-7.7\%$ & $3.824$ & $24.9\%$ & $3.275$ & $35.6\%$ & $2.249$ & $8.6\%$ & $0.953$ & $61.3\%$ & $1.171$ & $52.4\%$\\
$50000$ & $5.348$ & $-5.1\%$ & $3.624$ & $28.8\%$ & $3.274$ & $35.7\%$ & $2.590$ & $-5.2\%$ & $0.943$ & $61.7\%$ & $0.928$ & $62.3\%$\\
$100000$ & $5.050$ & $0.8\%$ & $3.668$ & $27.9\%$ & $3.495$ & $31.3\%$ & $2.529$ & $-2.8\%$ & $1.083$ & $56.0\%$ & $1.046$ & $57.5\%$\\
$1000000$ & $5.076$ & $0.2\%$ & $3.686$ & $27.6\%$ & $3.247$ & $36.2\%$ & $2.450$ & $0.5\%$ & $1.071$ & $56.5\%$ & $0.945$ & $61.6\%$\\
\hline
				\multicolumn{13}{|c|}{$\alpha=0.5$}\\
				\hline
&	\multicolumn{6}{c}{$c_{N,2}(n_0=0, \Delta\,t)=5.089\times 10^{-3}$}			&			\multicolumn{6}{|c|}{$c_{N,2}c_{N,n}(n_0=0, \Delta\,t)=2.461\times 10^{-3}$}\\
\hline
& \multicolumn{2}{c}{\text{Exact}}	&	\multicolumn{2}{c|}{\text{Approximation 1}}	&	\multicolumn{2}{c}{\text{Approximation 2}}	&	\multicolumn{2}{|c|}{\text{Exact}} & \multicolumn{2}{c|}{\text{Approximation 1}} &	\multicolumn{2}{c|}{\text{Approximation 2}} \\ \hline	
$N$	& \text{MC} &	\text{err} \%	& \text{MC} &	\text{err} \%	 & \text{MC}	& \text{err} \%	& \text{MC} &	\text{err} \%	& \text{MC}	& \text{err} \%	& \text{MC}	& \text{err} \%\\		
\hline	
$1000$ & $5.710$ & $-12.2\%$ & $5.651$ & $-11.0\%$ & $2.320$ & $54.4\%$ & $2.742$ & $-11.4\%$ & $6.369$ & $-158.8\%$ & $0.177$ & $92.8\%$\\
$10000$ & $5.633$ & $-10.7\%$ & $4.585$ & $9.9\%$ & $3.445$ & $32.3\%$ & $3.073$ & $-24.8\%$ & $1.954$ & $20.6\%$ & $1.388$ & $43.6\%$\\
$20000$ & $5.430$ & $-6.7\%$ & $4.393$ & $13.7\%$ & $3.635$ & $28.6\%$ & $2.678$ & $-8.8\%$ & $1.478$ & $39.9\%$ & $1.353$ & $45.0\%$\\
$50000$ & $5.260$ & $-3.4\%$ & $4.075$ & $19.9\%$ & $3.350$ & $34.2\%$ & $2.616$ & $-6.3\%$ & $1.128$ & $54.2\%$ & $1.038$ & $57.8\%$\\
$100000$ & $5.050$ & $0.8\%$ & $4.034$ & $20.7\%$ & $3.400$ & $33.2\%$ & $2.367$ & $3.8\%$ & $1.055$ & $57.1\%$ & $1.125$ & $54.3\%$\\
$1000000$ & $5.104$ & $-0.3\%$ & $4.043$ & $20.6\%$ & $3.245$ & $36.2\%$ & $2.467$ & $-0.3\%$ & $1.124$ & $54.3\%$ & $0.979$ & $60.2\%$\\
\hline
				\multicolumn{13}{|c|}{$\alpha=0.3$}\\
				\hline
&	\multicolumn{6}{c}{$c_{N,2}(n_0=0, \Delta\,t)=5.089\times 10^{-3}$}			&			\multicolumn{6}{|c|}{$c_{N,2}c_{N,n}(n_0=0, \Delta\,t)=2.461\times 10^{-3}$}\\
\hline
& \multicolumn{2}{c}{\text{Exact}}	&	\multicolumn{2}{c|}{\text{Approximation 1}}	&	\multicolumn{2}{c}{\text{Approximation 2}}	&	\multicolumn{2}{|c|}{\text{Exact}} & \multicolumn{2}{c|}{\text{Approximation 1}} &	\multicolumn{2}{c|}{\text{Approximation 2}} \\ \hline	
$N$	& \text{MC} &	\text{err} \%	& \text{MC} &	\text{err} \%	 & \text{MC}	& \text{err} \%	& \text{MC} &	\text{err} \%	& \text{MC}	& \text{err} \%	& \text{MC}	& \text{err} \%\\		
\hline	
$1000$ & $5.545$ & $-9.0\%$ & $4.008$ & $21.2\%$ & $3.118$ & $38.7\%$ & $2.008$ & $18.4\%$ & $0.268$ & $89.1\%$ & $0.290$ & $88.2\%$\\
$10000$ & $4.702$ & $7.6\%$ & $4.576$ & $10.1\%$ & $3.322$ & $34.7\%$ & $2.742$ & $-11.4\%$ & $0.930$ & $62.2\%$ & $0.776$ & $68.4\%$\\
$20000$ & $5.421$ & $-6.5\%$ & $4.471$ & $12.1\%$ & $3.255$ & $36.0\%$ & $2.195$ & $10.8\%$ & $0.797$ & $67.6\%$ & $0.703$ & $71.5\%$\\
$50000$ & $5.329$ & $-4.7\%$ & $4.472$ & $12.1\%$ & $3.179$ & $37.5\%$ & $2.671$ & $-8.5\%$ & $1.431$ & $41.9\%$ & $0.663$ & $73.1\%$\\
$100000$ & $5.050$ & $0.8\%$ & $4.477$ & $12.0\%$ & $3.208$ & $37.0\%$ & $2.513$ & $-2.1\%$ & $1.270$ & $48.4\%$ & $0.875$ & $64.4\%$\\
$1000000$ & $5.099$ & $-0.2\%$ & $4.370$ & $14.1\%$ & $3.266$ & $35.8\%$ & $2.485$ & $-1.0\%$ & $1.144$ & $53.5\%$ & $1.010$ & $59.0\%$\\

\hline
				\multicolumn{13}{|c|}{$\alpha=0.3$}\\
				\hline
&	\multicolumn{6}{c}{$c_{N,2}(n_0=0, \Delta\,t)=5.089\times 10^{-3}$}			&			\multicolumn{6}{|c|}{$c_{N,2}c_{N,n}(n_0=0, \Delta\,t)=2.461\times 10^{-3}$}\\
\hline
& \multicolumn{2}{c}{\text{Exact}}	&	\multicolumn{2}{c|}{\text{Approximation 1}}	&	\multicolumn{2}{c}{\text{Approximation 2}}	&	\multicolumn{2}{|c|}{\text{Exact}} & \multicolumn{2}{c|}{\text{Approximation 1}} &	\multicolumn{2}{c|}{\text{Approximation 2}} \\ \hline	
$N$	& \text{MC} &	\text{err} \%	& \text{MC} &	\text{err} \%	 & \text{MC}	& \text{err} \%	& \text{MC} &	\text{err} \%	& \text{MC}	& \text{err} \%	& \text{MC}	& \text{err} \%\\		
\hline	
$1000$ & $4.659$ & $8.4\%$ & $3.891$ & $23.6\%$ & $2.061$ & $59.5\%$ & $2.942$ & $-19.5\%$ & $0.182$ & $92.6\%$ & $0.002$ & $99.9\%$\\
$10000$ & $4.975$ & $2.2\%$ & $4.191$ & $17.6\%$ & $2.909$ & $42.8\%$ & $2.833$ & $-15.1\%$ & $0.222$ & $91.0\%$ & $0.400$ & $83.8\%$\\
$20000$ & $4.848$ & $4.7\%$ & $4.651$ & $8.6\%$ & $3.014$ & $40.8\%$ & $2.678$ & $-8.8\%$ & $0.876$ & $64.4\%$ & $0.385$ & $84.4\%$\\
$50000$ & $4.970$ & $2.3\%$ & $4.773$ & $6.2\%$ & $3.102$ & $39.0\%$ & $2.616$ & $-6.3\%$ & $1.208$ & $50.9\%$ & $0.431$ & $82.5\%$\\
$100000$ & $5.022$ & $1.3\%$ & $4.718$ & $7.3\%$ & $3.259$ & $36.0\%$ & $2.367$ & $3.8\%$ & $1.013$ & $58.8\%$ & $0.706$ & $71.3\%$\\
$1000000$ & $5.101$ & $-0.2\%$ & $4.843$ & $4.8\%$ & $3.282$ & $35.5\%$ & $2.419$ & $1.7\%$ & $1.271$ & $48.3\%$ & $1.108$ & $55.0\%$\\

\hline
		\end{tabular}
		}
	\caption{MC-estimated cumulants multiplied by $1000$ and comparison to their true values. $\left(b, \sigma, \nu
\right) = \left(5, 0.3, 2.5\right)$, $\Delta t=1/12$.}
	\label{tab:exact:1:12}
\end{table}
\clearpage

\section{Financial Applications}\label{sec:fin:app}
The day-ahead (also called spot) price of power or gas and in general of commodities exhibit mean-reversion, seasonality and spikes, this last feature is particularly difficult to be modeled in a pure Gaussian world. Different approaches have been investigated in order to somehow extend the classical Gaussian framework introduced in Lucia and Schwarz\mycite{LS02} and Schwartz and Smith\mycite{SchwSchm00}. Among others, Cartea and Figueroa\mycite{CarteaFigueroa}, Meyer-Brandis and Tankov\mycite{MBT2008} and Sabino  and Cufaro Petroni \mycite{cs20_2} have studied mean-reverting jump-diffusions to model sudden spikes, whereas Cummins et al.\mycite{CKM17, CKM18}, and Sabino\mycite{Sabino21a, Sabino20b} have considered variance gamma and in general tempered stable and CGMY processes to price power or gas derivative contracts. 

In this section we consider a setting similar to Benth and Benth\mycite{BenthBenth04} and Benth et al.\mycite{BKM07} where the evolution of the spot dynamics is driven by a OU-SNTS process and eventually by a second factor. Based on the results presented in Section\myref{sec:ou:ncts}, compared to Benth and Benth\mycite{BenthBenth04} and Benth et al.\mycite{BKM07} we provide non-arbitrage conditions with no numerical approximations and then we illustrate the effectiveness our findings in the pricing energy derivatives, namely a strip of call options using the FFT-based technique of Carr and Madan\mycite{Carr1999OptionVU}, a forward start Asian option with MC simulations and a swing option with the Least-Squares Monte Carlo (LSMC) presented in Hambly et al.\mycite{HHM11}. 
Finally, we describe how these findings can be adapted to forward markets as well in order to capture the Samuelson effect in spirit of the works of Benth et al.\mycite{BPV19}, Latini et al.\mycite{LPV19} or Piccirilli et al.\mycite{PSV20}.

\subsection{Call Options}\label{sec:fin:app:call}

Consider a spot price dynamics driven by the following one factor process
	\begin{equation}\label{eq:market}
		S(t) = F(0,t)\, e^{h(t) + N(t)}
	\end{equation}
	where $h(t)$ is a deterministic function, $F(0,t)$ is the forward curve derived from quoted products and reflects the seasonality, whereas $N(t)$ is the OU-SNTS process in\refeqq{eq:sol:OUNCTS_N}; as already mentioned, this market model basically coincides with that discussed in Benth and Benth\mycite{BenthBenth04} and Benth et al.\mycite{BKM07} when $\alpha=1/2$. 

Using  Lemma 3.1 in Hambly
et al.\mycite{HHM11},  the risk-neutral conditions are met when the deterministic function $h(t)$ is consistent with forward curve such that
\begin{equation}\label{eq:rn:spot}
	h(t) = -m_N(1, t)
\end{equation}
where $m_N(s, t)$ is the \cgf\ of $N(t)$, $m_N(s, t) = s\,N(0)e^{-bt} + m_{Z_Y}(s, t)$ and $m_{Z_Y}(s, t)$ is given by\refeqq{eq:mgf:ou:sncts} therefore, assuming for simplicity $N(0)=0$, we have
\begin{eqnarray}\label{eq:rn:ou:ncts}
	h(t) &=& \frac{c\,\beta^{\alpha}\,\Gamma(1-\alpha)}{\alpha\,b}\left[\tilde{J}\left(1, \sigma, \alpha, \beta,\frac{\beta}{\omega}\right) + \log \omega\right]
\end{eqnarray}
with $\frac{\sqrt{2\,\beta}}{\sigma}>1$ or $\sqrt{\frac{2\,(1-\alpha)}{\sigma^2\,\nu}}>1$.
Taking $\alpha=1/2$, the integrals can be written in terms of the logarithmic function, indeed
\begin{eqnarray*}
\tilde{J}\left(1, \frac{1}{2}, \beta_1, \beta_2\right)&=&  \int_{\beta_1}^{\beta_2} z^{-\frac{3}{2}}\left(z - \frac{\sigma^2}{2}\right)^{\frac{1}{2}}dz = 
\int_{\frac{2\,\beta_1}{\sigma^2}}^{\frac{2\,\beta_2}{\sigma^2}} y^{-\frac{3}{2}}\left(y - 1\right)^ {\frac{1}{2}}dy
\end{eqnarray*}
where we used the change of variable $z=\frac{\sigma^2\,y}{2}$, in addition using equation (35) in Sabino\mycite{Sabino21a} we have
\begin{eqnarray}
\tilde{J}\left(1, \frac{1}{2}, \beta_1, \beta_2\right)&=&
2\,\left(\log \left(\sqrt{\frac{2\,\beta_2}{\sigma^2}} + \sqrt{\frac{2\,\beta_2}{\sigma^2} - 1}\right) - 
\sqrt{\frac{2\,\beta_2 - \sigma^2}{2\,\beta_2}}\right) - \nonumber\\
&&2\,\left(\log \left(\sqrt{\frac{2\,\beta_1}{\sigma^2}} + \sqrt{\frac{2\,\beta_1}{\sigma^2} - 1}\right) - 
\sqrt{\frac{2\,\beta_1 - \sigma^2}{2\,\beta_1}}\right)
\label{eq:rn:integrals_p_0.5}.
\end{eqnarray} 

We consider a daily strip of $M$ call options with maturity $T$ and strike $K$ namely, a contract with payoff
\begin{equation*}
C(K, T) = \sum_{m=1}^M(S(t_m)-K)^+=\sum_{m=1}^M c_m(K, t_m), \quad t_1, t_2, \dots t_M=T.
\end{equation*}
Such a contract normally encompasses monthly, quarterly and yearly maturities, but is not very liquid and  is generally offered by brokers.

As a consequence of Proposition\myref{prop:chf:ou:sncts}, we can make use of the explicit form of the \chf\ of the of $\log\,S(t)=\log F(0,t) + h(t) + X(t)$ to compute the price of such a strip of calls using the FFT-based technique of Carr and Madan\mycite{Carr1999OptionVU}.

For this specific example we consider a set with realistic parameters,  $(b, \nu, \sigma) = (10, 0.7, 0.2)$, similar to that of Cummins et al.\mycite{CKM17} and Gardini et al.\mycite{Gardini20b} and let $\alpha$ vary. In addition, in order to better highlight the dependency on the maturity and $\alpha$, we take a flat forward curve $F(0,t) = 20$, indeed the seasonality changes the moneyness of the strip of call options and could shadow some features of the price model;  therefore we also assume an at-the-money strike $K=20$. 

Table\myref{tab:call} shows the numerical results relatively to different $\alpha$'s and different maturities spanning from one month to one year. Table\myref{tab:call} also compares the results obtained with the FFT method and with the MC method with $10^5$ simulations. As far as this last method is concerned, we also report the standard errors defined as the sample standard deviation divided by the square root of the relative number of simulations. 

As expected, we observe that the FFT and the MC approaches return comparable results, of course the FFT method is known to be faster. Moreover, as also mentioned in Subsection\myref{sub:sec:ou:sncts:alg}, except for $\alpha=0.5$, which corresponds to an IG law, the simulation relies on the acceptance rejection method and is thus more computationally expensive.
Moreover, the call options prices exhibit a non-negligible monotonic and increasing dependence on the stability parameter $\alpha$. Taking the values for $\alpha=0.5$ (OU-NIG) as the benchmark, given the used set of remaining parameters, the difference is of circa plus or minus fifteen percent. On the other hand, Figure\myref{fig:ou:ncts:call:strike} shows that the price difference is non-negligible for almost at-the-money options, whereas it disappears for deep out and in-the-money options. Besides these observations, an important fact is visible in Figure\myref{fig:ou:ncts:call:maturity} which shows how each term of the strip depends on the time-to-maturity for an at the money-option: apparently, the time-to-maturity does not to play a relevant role for values above one month. The cause of this evident drawback is the fact that we are using a one-factor model and we remark that we can identify this issue because we have switched the seasonality off. Indeed, even though we are using a lower value for the mean reversion level $b$ than those presented in other studies (see for instance once again Benth and Benth\mycite{BenthBenth04}, Benth et al.\mycite{BKM07}), such a value is large enough to ``tail off'' the long term variability, hence a second factor is required for the pricing of contracts with long maturities. In the following section, we will illustrate the calibration and application of a two-factor model to the pricing of a forward start Asian option and a swing option. Of course, the FFT and MC methods discussed in this section can be easily adapted to a two-factor price dynamics as well.

\begin{table}[ht!]
    \centering
        \resizebox{\textwidth}{!}{
	\begin{tabular}{*{11}{|c|c|c|c|c|c|c|c|c|c|c}}
				\hline
				 & \multicolumn{2}{|c|}{$\alpha=0.1$} & \multicolumn{2}{|c}{$\alpha=0.3$} & \multicolumn{2}{|c|}{$\alpha=0.5$} & \multicolumn{2}{|c|}{$\alpha=0.7$} & \multicolumn{2}{|c|}{$\alpha=0.9$}\\
				\hline
				& FFT & MC & FFT & MC & FFT & MC & FFT & MC & FFT & MC
				\\
				\hline
$T = 1/12$ & $3.3259$ & $3.281 \pm 0.0335$ & $3.8392$ & $3.7348 \pm 0.034$ & $4.5342$ & $4.3131 \pm 0.0334$ & $5.3861$ & $5.1903 \pm 0.0324$ & $6.5152$ & $6.31 \pm 0.0319$\\
$T = 3/12$ & $16.481$ & $16.294 \pm 0.104$ & $18.017$ & $18.241 \pm 0.106$ & $19.905$ & $19.569 \pm 0.105$ & $22.184$ & $22.115 \pm 0.102$ & $25.308$ & $25.328 \pm 0.098$\\
$T = 1/2$ & $38.078$ & $37.5626 \pm 0.1738$ & $40.881$ & $40.405 \pm 0.179$ & $44.305$ & $43.853 \pm 0.177$ & $48.515$ & $48.691 \pm 0.173$ & $54.445$ & $54.868 \pm 0.166$\\
$T = 8/12$ & $59.749$ & $59.069 \pm 0.2242$ & $63.799$ & $64.553 \pm 0.231$ & $68.745$ & $68.227 \pm 0.229$ & $74.878$ & $75.248 \pm 0.224$ & $83.606$ & $84.387 \pm 0.216$\\
$T = 1$ & $81.421$ & $80.9306 \pm 0.2655$ & $86.716$ & $86.762 \pm 0.274$ & $93.186$ & $92.575 \pm 0.271$ & $101.24$ & $101.79 \pm 0.27$ & $112.77$ & $112.94 \pm 0.26$\\
\hline
		\end{tabular}
		}
	\caption{Strip of at-the money call option prices obtained with the FFT and the MC method with $10^5$ simulations under the assumption that $\left(b, \sigma, \nu
\right) = \left(10, 0.2, 0.7\right)$.}
	\label{tab:call}
\end{table}				

\begin{figure}
        \begin{subfigure}[c]{.5\textwidth}{
                \includegraphics[width=70mm]{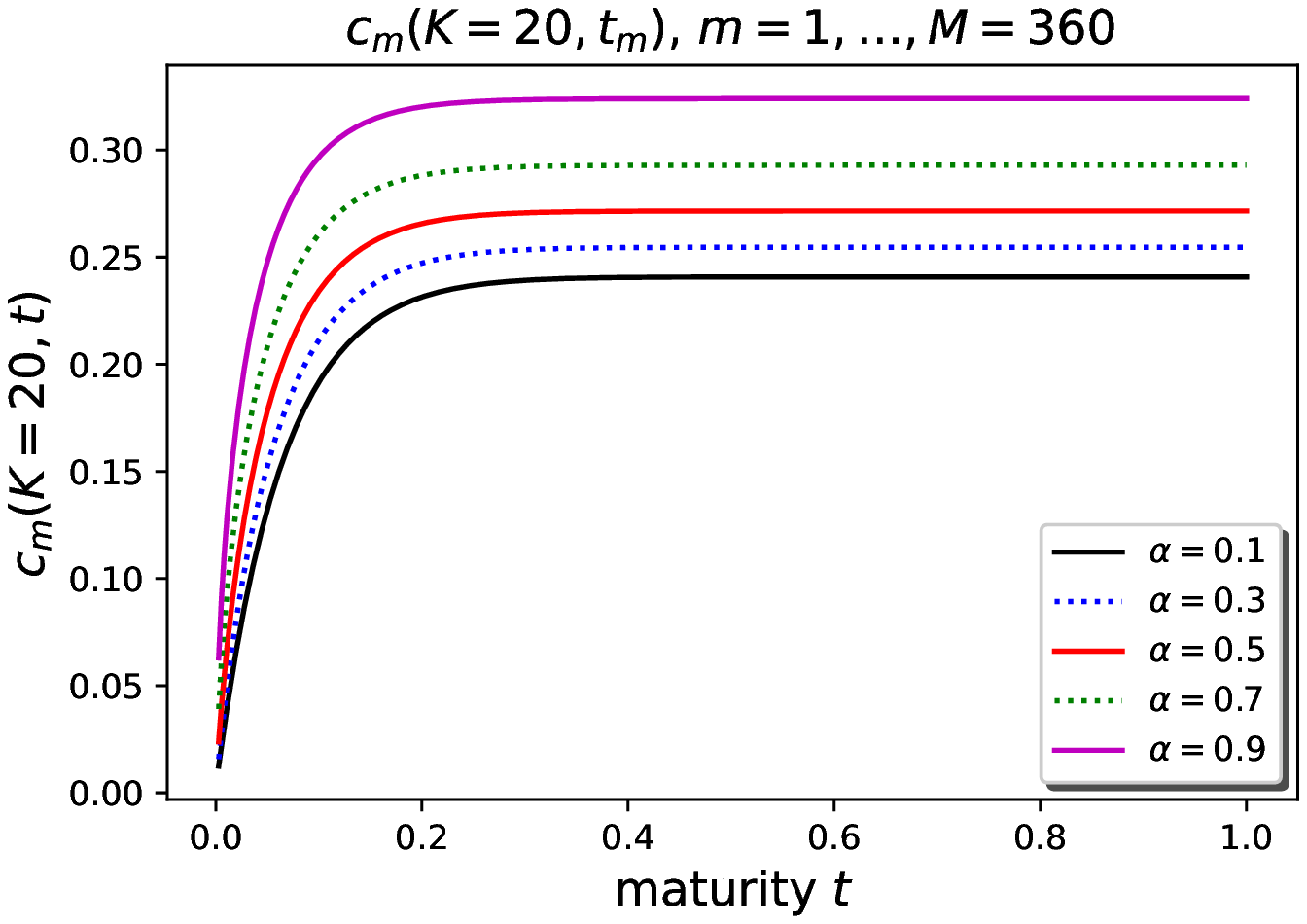}
                }
				\caption{$c_m(K, t_m), m=1, \dots, 360$}\label{fig:ou:ncts:call:maturity}								
        \end{subfigure}
        \begin{subfigure}[c]{.5\textwidth}{
                \includegraphics[width=70mm]{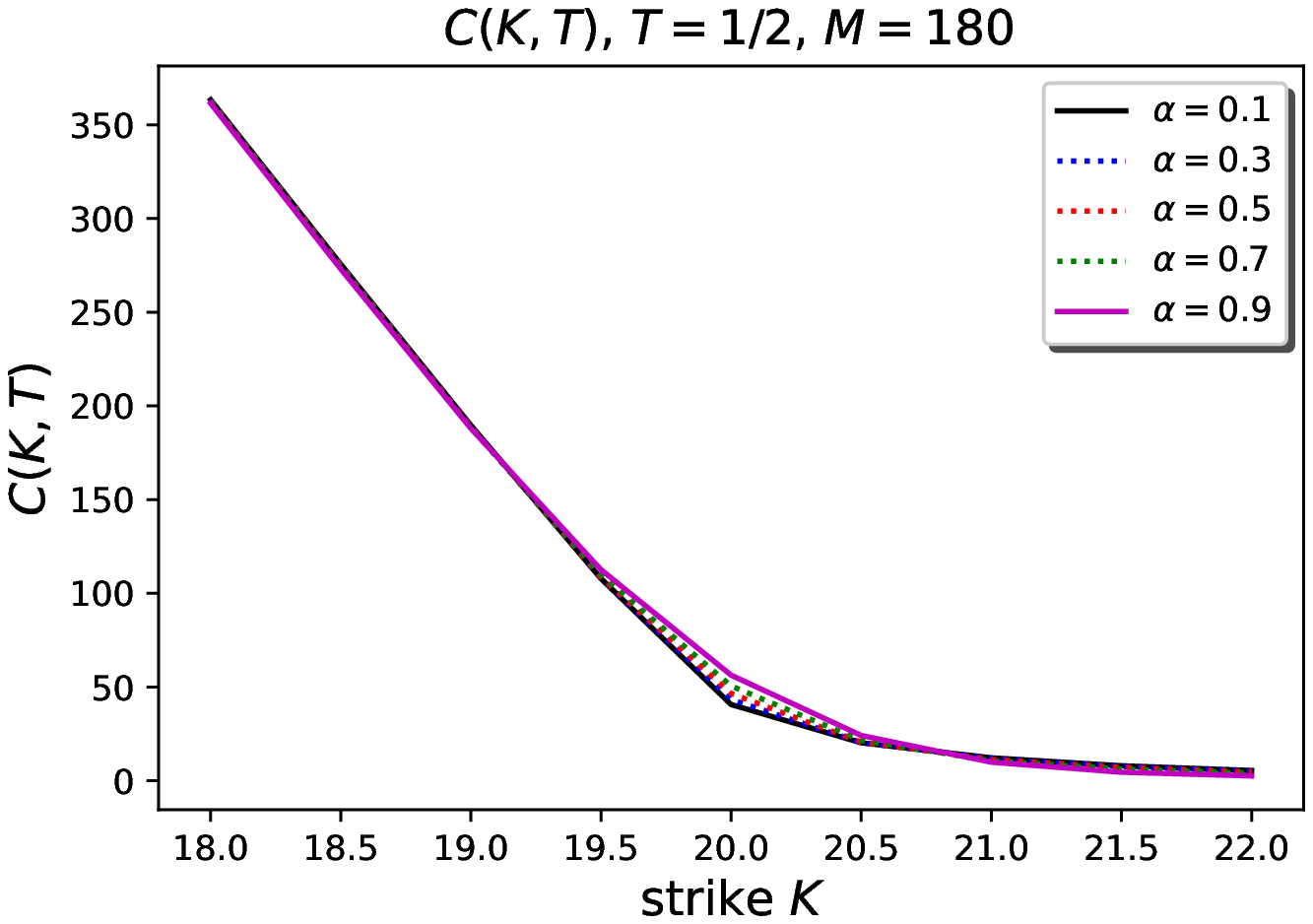}
                }
				\caption{Effect on the strike}\label{fig:ou:ncts:call:strike}
        \end{subfigure}
\label{fig:ou:call}				
\end{figure}

\subsection{A Two-factor Model}\label{sec:fin:app:two:factors}
The previously defined one-factor model\refeqq{eq:market} can easily be turned into a two-factor one, adding a second  process, for instance a NIG, a Variance Gamma or any another NTS process,  obtaining a normal tempered stable version of the two factors Gaussian model of Schwartz and Smith\mycite{SchwSchm00}. For the reasons illustrated in the previous section, such a model is of course more realistic, the two factors are often associated by practitioners to a short and a long term variability. We consider the following  evolution of the spot prices
\begin{equation}\label{eq:market:two:factors}
		S(t) = F(0,t)\, e^{h_1(t) + h_2(t) + N_1(t) + N_2(t)}
	\end{equation}
	where $h_1(t)$, $N_1(\cdot)$ and $F(0,t)$ play the same role of as the terms in\refeqq{eq:market}, whereas $N_2(\cdot)$ is an independent NTS process of the type of Equation\refeqq{eq:sub:bm}, and $h_2(t)$, similarly to $h_1(t)$,  is a deterministic function related to the \che\ of $N_2(\cdot)$ given by\refeqq{eq:che:ncts}. We denote  $(\alpha_1, \sigma_1, \nu_1, \theta_1, b)$ and $(\alpha_2, \sigma_2, \nu_2, \theta_2)$ the parameters of the two processes respectively. Further, we adopt the point of view of de Jong and Schneider\mycite{DJS09} and assume that the process $N_2(\cdot)$ is related to the month-ahead forward price $F_{MA}(t)$ at time $t$ obtained rolling the front-month forward once expired. Writing then 
\begin{equation}
F_{MA}(t) = F_{MA}(0) e^{h_2(t) + N_2(t)}
\label{eq:front:month}
\end{equation}
has the advantage to relate the process $N_2(\cdot)$ to observable prices and at the same time to simplify the parameter estimation of\refeqq{eq:market:two:factors}. 

The unknown parameters $(\sigma_2, \nu_2, \theta_2)$ can be estimated with the generalized method of moments or with simulation based approaches because, except for $\alpha=0.5$ (the NIG process), the transition density of $N_2(\cdot)$ is quite complicated and one cannot always rely on the maximum-likelihood method. In the following, we simplify the framework and take $\alpha_1=\alpha_2=0.5$, namely, we assume that $N_1(\cdot)$ is a OU-NIG process, whereas $N_2(\cdot)$ is a plain NIG process. Such a simplification allows us to accomplish the parameter estimation with the maximum likelihood method implemented by the package \emph{GeneralizedHyperbolic} in R. Nevertheless, in the following section, we will illustrate the sensitivity of the Asian option value to different parameters $\alpha$. We consider the historical data of the day-ahead and month-ahead forward prices of the German gas market NCG in the four years time window from January 1st 2016 to December 31st 2019 for a total of $1461$ values.

After calibrating the parameters of the process $N_2(\cdot)$ we proceed with the estimation of those of the process $N_1(\cdot)$, referred as a short term factor, removing the  seasonality component from the day-ahead time-series at first, using a decomposition with a double-cosine plus a linear trend (see Figure\myref{fig:ncg:detrending}\footnote{We use the package \emph{STL} in R}).  
Taking then the logarithm of each value of the deseasonalized time-series of the day-ahead and of the month-ahead NCG prices denoted $s_k$ and $f_k$, $k=1, \dots, 1461$, respectively, we have
\begin{equation}
s_{k+1} - f_{k+1} = (s_k - f_k)\,e^{-b\,\Delta t} + \epsilon_{k+1}
\label{eq:ou:calibration}
\end{equation}
where $\Delta t = 1/365$ (one day) and $\epsilon_{k}, k=1, \dots$ is distributed according to the law of $Z_Y(\Delta t)$ of\refeqq{eq:sol:OUNCTS_Z} when $Y(\cdot)$ is a NIG process. From the sample statistics of the residual it turns out that their distribution is symmetric, thus can rely on the modeling via a symmetric OU-NIG process and we can neglect the parameter $\theta_1$. We remark that due to Proposition\myref{prop:rv:ou:sncts}, $\epsilon_{k}, k=1, \dots$  is not a NIG distributed random sequence but rather the convolution of a NIG law plus a compound Poisson law. On the other hand, because of the discussion in Section\myref{sub:sec:ou:sncts:num}, it is reasonable to neglect the latter contribution if  $\Delta t$ is small, as it is in our case, so that one can accomplish the parameter estimation with the maximum-likelihood method. However, as observed once again in Section\myref{sub:sec:ou:sncts:num}, Approximation 1  is more reliable than the plain Euler scheme, and therefore, we adopt such a choice.
For sake of completeness, the time-series of the month-ahead forward contract is available on trading days only, whereas that of the day-ahead is available every day. In order to overcome this inconsistency, we create fictitious values for the month-ahead time-series on non-trading days by a random sampling from a NIG law with the estimated parameters $(\hat{\sigma}_2, \hat{\nu}_2, \hat{\theta}_2)$.

Table\myref{tab:parameters} reports the results of the estimation procedure, whereas Figures\myref{fig:ncg:ma} and\myref{fig:ncg:da} show that the NIG distribution fit the data of both the month-ahead and the short term component of the day-ahead series quite well. Nevertheless, the approximation that we have adopted in the calibration of the parameters of the latter factor is not always reliable in the context of the pricing of energy derivatives especially if one deals with forward start contracts. To this end, in order to avoid to introduce an intrinsic bias, one must implement the simulation of the skeleton of the OU-NIG component (and OU-NTS) with the exact simulation scheme or at least select such a procedure for those time intervals that are larger than one day. In the following application, we adopt such a strategy for the pricing of a forward start Asian option and a swing option. 

\begin{table}[ht!]
	\centering
		\begin{tabular}{*{6}{|c|c|c|c|c|c}}
		\hline
			$\hat{\sigma_1}$ & $\hat{\nu_1}$ & $\hat{b}$ & $\hat{\sigma_2}$ & $\hat{\nu_2}$ & $\hat{\theta_2}$\\
			\hline
			$0.2835$ & $0.0804$ & $39.86$ & $0.3142$ & $0.1023$ &-0.019\\
			\hline
		\end{tabular}
	\scriptsize
	\caption{\footnotesize{Estimated Parameters}}\label{tab:parameters}
\end{table}	
\begin{figure}
        \begin{subfigure}[c]{.3\textwidth}{
                \includegraphics[width=50mm]{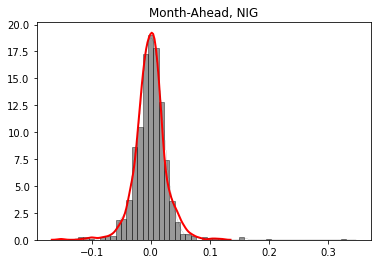}
                }
					\caption{Month-ahead}\label{fig:ncg:ma}			
				\end{subfigure}\,\,\,
        \begin{subfigure}[c]{.3\textwidth}{
                \includegraphics[width=50mm]{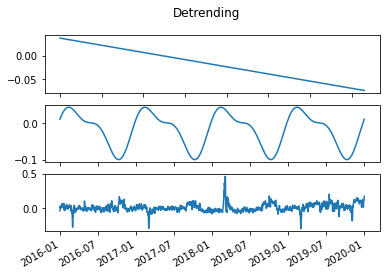}
                }
					\caption{Detrend}\label{fig:ncg:detrending}								
        \end{subfigure}\,\,\,
        \begin{subfigure}[c]{.3\textwidth}{
                \includegraphics[width=50mm]{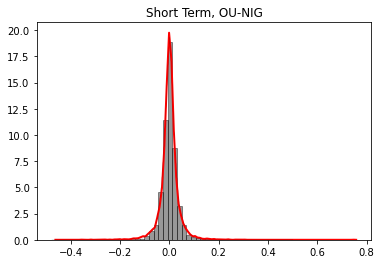}
                }
					\caption{Short-term}\label{fig:ncg:da}								
        \end{subfigure}
\label{fig:calibration}				
\end{figure}

\subsection{Asian Options}\label{sec:fin:app:asian}
As a second financial application we consider the pricing of Asian options with MC simulations in a two-factor market model using the calibration of Section\myref{sec:fin:app:two:factors} and let also $\alpha$ vary. Of course, the parameter estimation carried out in Section\myref{sec:fin:app:two:factors} assumed that the two factors were both NIG distributed with $\alpha=0.5$, one might change $\alpha$ for each factor. In addition, one might argue that for a general NTS distributed factor the remaining parameters should also be different that those with $\alpha=0.5$. Nevertheless, we only let the stability parameter change for both the two factors and study its impact on the option price.

It is worthwhile noting that knowing the \chf's of the two factors of the market model, one could also price Asian options with FFT methods (see Zhang and C. Oosterlee\mycite{ZhangOosterlee13_a}) that are faster than MC methods. On the other hand, the former approach  provides a view on the distribution of the potential cash-flows of derivative contracts giving a precious information to risk managers or to trading units. 

We recall that the payoff at maturity $T$ of an Asian option with European style and strike price $K$ is
\begin{equation*}
	A(K, T) = \left(\frac{\sum_{i=1}^{I}}{I}S(t_i) - K\right)^+.
\end{equation*}
Our numerical example considers a forward start contract with $I=90$ that starts settling at $T_1=1/4$ with maturity $T=1/2$ and $K=11.5$. 

Since the forward start feature and the observations of Section\ref{sub:sec:ou:sncts:num}, we rely on the exact simulation scheme only for the OU-SNTS component because, even though the presence of the second factor could shadow the impact, one single large time point could create a bias in the option price. 

Figure\myref{tab:asian} shows the prices and the errors defined as the sample standard deviation divided by the number of scenarios $N$ changing the parameter $\alpha$. We can conclude that the MC procedure is convergent but that at least $N=5\times 10^4$ scenarios are required to attain a reliable price and that $\alpha$ has a non-negligible impact. Indeed, the price difference with two consecutive $\alpha$'s in Table\myref{tab:asian} is of circa ten percent that is a large number for an Asian option that settles over three months only. 

\begin{table}[ht!]
	\centering
	\resizebox{\textwidth}{!}{
		\begin{tabular}{*{11}{|c|c|c|c|c|c|c|c|c|c|c}}
		\hline
			& \multicolumn{2}{c|}{$\alpha=0.1$} & \multicolumn{2}{c|}{$\alpha=0.3$} & \multicolumn{2}{c|}{$\alpha=0.5$} & \multicolumn{2}{c|}{$\alpha=0.7$} & \multicolumn{2}{c|}{$\alpha=0.9$}\\
			\hline
			$N$ & price & error & price & error & price & error & price & error & price & error\\		
			\hline
			$1000$ & $0.830$ & $0.044$ & $0.834$ & $0.044$ & $0.847$ & $0.045$ & $0.822$ & $0.044$ & $0.855$ & $0.047$\\
			$10000$ & $0.800$ & $0.015$ & $0.801$ & $0.013$ & $0.808$ & $0.014$ & $0.792$ & $0.014$ & $0.810$ & $0.014$\\
			$20000$ & $0.7812$ & $0.0100$ & $0.7894$ & $0.0010$ & $0.7989$ & $0.0099$ & $0.8036$ & $0.0100$ & $0.8154$ & $0.0099$\\
			$50000$ & $0.7719$ & $0.0063$ & $0.7814$ & $0.0063$ & $0.7906$ & $0.0062$ & $0.8011$ & $0.0063$ & $0.8096$ & $0.0062$\\
			$100000$ & $0.7703$ & $0.0043$ & $0.7806$ & $0.0043$ & $0.7909$ & $0.0044$ & $0.8008$ & $0.0044$ & $0.8114$ & $0.0044$\\
			\hline
		\end{tabular}
	}
	\scriptsize
	\caption{\footnotesize{Forward start Asian option.}}\label{tab:asian}
\end{table}

\subsection{Swing Options}\label{sec:fin:app:swing}
As a last application we consider the pricing of swing options that are typical components of gas contracts, which offer the opportunity to vary the contracted volume under a number of restrictions. We consider the day-ahead market model illustrated in Section\myref{sec:fin:app:two:factors} and its relative estimated parameters, however, in contrast to the previous example, we take NIG-based factors as per the original calibration only.

Let the maturity date $T$ be fixed and the payoff at time $t<T$ be given by $(S(t)-K)^+$
where  $K$ denotes the strike price, we also assume that only one unit of the underlying can be exercised any
time period. Let $V(n, s, t)$ denote the price of a swing option at time $t$ given the spot price $s$
which has $n$ out of $N$ exercise rights left. For $m=1,\dots, M_S$, the dynamic programming principle allows us to
write 
\begin{equation}
V(n, s, t_m) =\max
              \left\{
                   \begin{array}{ll}
									 \EXP{V(n, S(t_{m+1}), t_{m+1})|S(t_m) = s}, \\
                                 \\
									\EXP{V(n, S(t_{m+1}), t_{m+1})|S(t_m) = s} + (s-K)^+
                   \end{array}
             \right\},\quad n < N 
\label{eq:swing}
\end{equation}
and $V(n, s, T) = (S(T) - K)^+$, $n \le N$ and $V(0, s, t)=0$.
We solve the recursion equation with the modified
version of the LSMC, introduced in
Longstaff-Schwartz\mycite{LSW01}, detailed in Hambly et al\mycite{HHK09} where the continuation value is approximated with a linear regression using power polynomials with  $m=1,\dots, M_S$
\begin{equation*}
	\EXP{V(n,  S(t_{m+1}), t_{m+1})|S(t_m) = s}\simeq a_0 + a_1 S(t_m) + \dots, +a_B S^B(t_m), \quad n < N.
\end{equation*}
We used $B=3$, but the regression may be performed on a different set of basis functions as well.

Several other approaches have been proposed: for instance, one may solve the recursion by adapting the method of Ben-Ameur et al.\mycite{BBKL2007} or might use the quantization technique of Bardou et al.\mycite{BBP07}. In alternative, one can also use the tree method of Jaillet et al.\mycite{JRT04} or the Fourier cosine expansion in Zhang and C. Oosterlee\mycite{ZhangOosterlee13_b} taking advantage of the explicit form of the \chf's of the two factors. 

Table\myref{tab:swing} shows the values and MC errors relatively to the pricing of a $120-120$ forward start swing option with maturity $T=1 + 1/3$ and strike price $K=11.5$, namely the holder has $N=120$ rights and must exercise all of them. We observe that the LSMC combined with Algorithm\myref{alg:sim:ou:sncts} for the OU-NIG factor - and a standard procedure for the generation of the NIG process- produces unbiased results and apparently $2\times 10^4$ simulations are required to attain an acceptable convergence.  Overall, it is evident that our newly developed approach can achieve high accuracy as well as efficiency.
\begin{table}[ht!]
	\centering
		\begin{tabular}{*{6}{|c|c|c|c|c|c}}
		\hline
			$N$ & $1000$ & $10000$ & $20000$ & $50000$ & $100000$\\
			\hline
			price & $451.1$ & $454.37$ & $457.90$ & $459.34$ & $459.69$\\
			error & $15.7$ & $5.20$ & $3.69$ & $2.33$ & $1.66$\\
			\hline
		\end{tabular}
	\scriptsize
	\caption{\footnotesize{$120-120$ swing option.}}\label{tab:swing}
\end{table}


\subsection{Application to Forward Markets}\label{sec:extensions}
The spectrum of financial applications of the results and the simulation procedures illustrated in Section\myref{sec:ou:ncts} is not restricted to the modeling of day-ahead commodity prices and to examples involving mean-reverting OU processes. For instance, it is common practice to model the evolution $f(t, T)$ of a forward contract  at time $t$ with delivery at $T$ as a geometric BM with a time-dependent volatility function in order captures the Samuelson effect. In particular, Kiesel et al.\mycite{KSB09} have considered the following two-factor market dynamics 
\begin{eqnarray*}
f(t, T) &=& f(0,T)\,\exp\left\{h_1(t, T) + h_2(t, T) + \sigma_1\int_0^te^{-b\,(T - u)}\,dW_1(u) + \sigma_2\,W_2(t)\right\}\\
 &=& f(0,T)\,\exp\left\{h_1(t, T) + h_2(t, T) + X_1(t, T) + X_2(t)\right\}
\end{eqnarray*}
where $W_1(\cdot)$ and $W_2(\cdot)$ are correlated or independent Brownian motions and $h_1(t, T)$ and $h_2(t, T)$ are two time-dependent functions that ensure that $f(\cdot, T)$ is a martingale. From such a model one can obtain day-ahead price evolution letting $t=T$. On the other hand, leaving the Gaussian framework by replacing the two Wiener processes by two general \Levy\  processes $L_1(\cdot)$ and $L_2(\cdot)$, one can rewrite the first factor as $X_1(t, T) = e^{-b\,T}\int_0^te^{-b\,(t - u)}\,dL_1(u) = e^{-b\,T}Z_1(t)$. Of course, when $L_1(\cdot)$ is  a SNTS process, $Z_1(\cdot)$ becomes the additive process studied in Section\myref{sec:ou:ncts}, and therefore we can easily adapt all the results obtained so far. 

The contract above is a fictitious contract often used to build the stochastic evolution of a traded contract $F(t, T_1, T_2)$, sometimes called swap in energy markets, at time $t$, maturity $T$, $t\le T\le T_1 < T_2$, and with delivery period $[T_1, T_2]$. In alternative, it is also common practice to directly model  $F(t, T_1, T_2)$ as
\begin{equation*}
F(t, T_1, T_2) = F(0, T_1, T_2)e^{\int_0^t\Gamma_1(u, T_1, T_2) dL_1(u) + \Gamma_2(T_1, T_2)\,X_2(t) - m_1(t, T_1, T_2) - m_2(t, T_1, T_2)}
\end{equation*}
where 
$m_1(t, T_1, T_2)$ and $m_2(t, T_1, T_2)$ are two deterministic functions that ensure that $F(\cdot, T_1, T_2)$ is a martingale.

Instead of an exponential setting, recently Piccirilli et al.\mycite{PSV20} have recently proposed a class of models, named Non-Overlapping-Arbitrage models (NOA), with the aim or capturing the Samuelson effect and reproducing the different levels and shapes of the implied volatility profiles displayed by options. In particular, they assume a NIG setting where the 
\begin{eqnarray}
	F(t, T_1, T_2) &=& F(0, T_1, T_2) + \int_0^t\Gamma_1(u, T_1, T_2) dL_1(u) + \Gamma_2(T_1, T_2)\,X_2(t) = \nonumber \\
	&& F(0, T_1, T_2) + X_1(t, T_1, T_2) + \Gamma_2(T_1, T_2)\,X_2(t),
\label{eq:noa}
\end{eqnarray}
where $L_1(\cdot)$ and $X_2(\cdot)$ are two independent centered NIG processes. In this specific case  
\begin{equation}
\Gamma_1(u, T_1, T_2) = \frac{\gamma_1}{b\,(T_2 - T_1)}\left(e^{-b\,(T_1 -u)} - e^{-b\,(T_2 -u)}\right).
\label{eq:gamma1}
\end{equation}
\begin{equation}
\Gamma(T_1, T_2) = \frac{1}{T_2 - T_1}\int_{T_1}^{T_2}\gamma(u)\,du
\label{eq:gamma2}
\end{equation}
are two deterministic functions that are meant to capture the Samuelson effect in option pricing (see also Jaeck and Lautier \mycite{JL16}). Indeed, in the spirit of Benth et al.\mycite{BPV19} and Latini et al.\mycite{LPV19}, the special form of the coefficients
arises from the implicitly underlying assumption that the swap can be written as the average
over an underlying artificial futures price with instantaneous delivery.

Of course, such models are related to the additive process studied in Section\myref{sec:ou:ncts}, because, after some algebra, it results
\begin{eqnarray*}
	 X_1(t, T_1, T_2) &=& \frac{\gamma_1}{b\,(T_2 - T_1)}\left(e^{-b\,(T_1 -T)} - e^{-b\,(T_2 -
	T)}\right) \int_0^t e^{-b\,(t -u)} dL_1(u)  \\
	&=&  \Gamma_1(T, T_1, T_2) Z(t), 
\end{eqnarray*}
hence the \chf\ and in the particular, the simulation procedure of the skeleton of the additive process $X_1(\cdot, T_1, T_2)$ can be derived from those of $Z(\cdot)$ under the assumption that $L_1(\cdot)$ is symmetric. It is worthwhile noticing that the Piccirilli et al. found an explicit form of the \chf\ of $X_1(t, T_1, T_2)$ when $L_1(\cdot)$ is a centered NIG process that can be extended to symmetric NTS processes based on Proposition\myref{prop:chf:ou:sncts}, we omit an explicit proof to avoid overloading the paper with lengthy details.   

We remark then that our simulation procedure enlarges the applicability of NOA models and other models based on the additive process $Z_1(\cdot)$ driven by NIG and in general NTS processes. Indeed, trading units rely on quoted options as a starting point to  calibrate the parameters of their models to then price other derivative contracts or to implement trading strategies. Some of them can be evaluated without MC methods however, simulation approaches are employed by risk managers to derive information about the distribution of the traded portfolios and a view on their risk profile. To this end, it is fundamental to rely on exact and non-biased simulation schemes with sufficiently fast computational speed that can be run with different parameter settings and market conditions.

\section{Concluding Remarks}\label{sec:conclusions}

In this study we have investigated the pricing of energy derivatives in markets driven by NTS processes that generalize NIG processes which are commonly used in the modeling of energy markets. To this end, the first contribution of our study is the full description of NTS processes of OU type, along with the derivation of characteristic function of their transition law in closed form. This result is instrumental to determine their statistical properties and compared to the current state of affairs (see for instance Benth and Benth\mycite{BenthBenth04} and Benth et al.\mycite{BKM07}) we derive non-arbitrage conditions for markets driven by such processes without resorting to numerical approximations or numerical integrations. 

The second contribution consists in defining an efficient algorithm for the exact simulation of the trajectories of symmetric OU-NTS and OU-NIG processes that is particularly suitable for forward start contracts, where the standard Euler scheme would return biased results. Indeed, even though it is common practice to rely on such an approximated scheme, we have shown that its performance quickly deteriorates. We also propose an alternative approximation scheme that apparently outperforms the Euler approximation having the same computational effort.

We have illustrated the applicability of these results to the pricing of three common derivative contracts in energy markets, namely a strip of daily call options, an Asian option with European style and a swing option.
In the first example we have considered a one-factor model purely driven by a single mean-reverting OU-NTS process. We have made use of the explicit knowledge of the \chf\ to implement the pricing with the FFT-based technique of Carr and Madan\mycite{Carr1999OptionVU}, and then have compared the outcomes to those obtained via MC simulations. We have noticed that such a model, although performing reasonably well in a wide range of applications  (see once again Benth and Benth\mycite{BenthBenth04} and Benth et al.\mycite{BKM07}), is not suitable for the evaluation of contracts with long maturities which drove us to add a second factor like in the setting proposed by de Jong and Schneider\mycite{DJS09} in the Gaussian framework.

In the second example, we have calibrated the two factor model taking the German day-ahead NCG prices and have evaluated a forward start Asian option. Based on our theoretical results, we have conceived a joint calibration-simulation strategy: under the assumption that one of the factors is a symmetric OU-NIG process we have suggested to estimate its parameters using the proposed approximation, whereas we have proposed to use the exact simulation scheme for the pricing of energy derivatives especially if they are forward start contracts. We also have studied the impact of moving from a NIG model to a general NTS model and have found that this has a non-negligible impact on the Asian option price.

In addition, we have shown that the proposed exact simulation algorithm, combined with the LSMC approach of Hambly et al\mycite{HHK09}, provides an efficient and accurate pricing of an one year $120-120$ swing option. Furthermore, our results are not restricted to OU processes and to the modeling of spot prices. Indeed, in the spirit of Benth et al.\mycite{BPV19}, Latini et al.\mycite{LPV19} and Piccirilli et al.\mycite{PSV20} they can be adapted to capture the Samuelson effect and to volatility smiles.

Finally, future studies could cover the extension to  a multidimensional
framework for instance adopting the view of Luciano and Semeraro\mycite{SL2010}, Ballotta and Bonfiglioli\mycite{BB2013} or the recent approaches of Gardini et al.\mycite{Gardini20b, Gardini20a} and Lu\mycite{lu2020}. 

We remark, that all the algorithms that we have discussed are based
on the sequential-forward generation of processes. On the other hand, as shown in Pellegrino and
Sabino\mycite{PellegrinoSabino15} and Sabino\mycite{Sabino20a}, the backward simulation is more suitable to the LSMC method and provides a computational advantage in the pricing of swings and storages. Therefore, a last topic deserving further investigation is the study of the simulation of OU-NTS processes backward in time.  

\section*{Acknowledgements}
I would like to express my gratitude to Prof. Simone Boehrer for her critical and constructive comments which helped to improve this article.
        \bibliographystyle{plain}
        \bibliography{biblioAll}

\begin{thebibliography}{10}

\bibitem{BB2013}
L.~Ballotta and E.~Bonfiglioli.
\newblock {Multivariate Asset Models Using {L}\'{e}vy Processes and
  Applications}.
\newblock {\em The European Journal of Finance}, 13(22):1320--1350, 2013.

\bibitem{BBP07}
O.~Bardou, S.~Bouthemy, and G.~Pag\'es.
\newblock Optimal {Q}uantization for the {P}ricing of {S}wing {O}ptions.
\newblock {\em Applied Mathematical Finance}, 16(2):183--217, 2009.

\bibitem{BJS1998}
O.~E. Barndorff-Nielsen, J.~L. Jensen, and M.~S{\o}rensen.
\newblock Some {S}tationary {P}rocesses in {D}iscrete and {C}ontinuous {T}ime.
\newblock {\em Advances in Applied Probability}, 30(4):989–1007, 1998.

\bibitem{BNSh01}
O.E. Barndorff-Nielsen and N.~Shephard.
\newblock Non-{G}aussian {O}rnstein-{U}hlenbeck-based {M}odels and some of
  their {U}ses in {F}inancial {E}conomics.
\newblock {\em Journal of the Royal Statistical Society: Series B},
  63(2):167--241, 2001.

\bibitem{BBKL2007}
H.~Ben-Ameur, M.~Breton, L.~Karoui, and P.~L'Ecuyer.
\newblock {A {D}ynamic {P}rogramming {A}pproach for {P}ricing {O}ptions
  {E}mbedded in {B}onds}.
\newblock {\em Journal of Economic Dynamics and Control}, 31(7):2212--2233,
  July 2007.

\bibitem{BKM07}
F.E. Benth, J.~Kallsen, and T.~Meyer-Brandis.
\newblock A non-gaussian ornstein-uhlenbeck process for electricity spot price
  modeling and derivatives pricing.
\newblock {\em Applied Mathematical Finance}, 14(2):153--169, 2007.

\bibitem{BDPL18}
F.E. Benth, L.~Di Persio, and S.~Lavagnini.
\newblock Stochastic {M}odeling of {W}ind {D}erivatives in {E}nergy {M}arkets.
\newblock {\em Risks, MDPI, Open Access Journal}, 6(2):1--21, 2018.

\bibitem{BPV19}
F.E. Benth, M.~Piccirilli, and T.~Vargiolu.
\newblock Mean-reverting {A}dditive {E}nergy {F}orward {C}urves in a
  {H}eath–{J}arrow–{M}orton {F}ramework.
\newblock {\em Mathematics and Financial Economics}, 13:543--577, 2019.

\bibitem{BenthBenth04}
F.E. Benth and J.~\v{S}altyt\'e Benth.
\newblock The {N}ormal {I}nverse {G}aussian {D}istribution and {S}pot {P}rice
  {M}odelling in {E}nergy {M}arkets.
\newblock {\em International Journal of Theoretical and Applied Finance},
  07(02):177--192, 2004.

\bibitem{Carr1999OptionVU}
P.~Carr and D.B. Madan.
\newblock Option {V}aluation {U}sing the {F}ast {F}ourier {T}ransform.
\newblock {\em Journal of Computational Finance}, 2:61--73, 1999.

\bibitem{CarteaFigueroa}
A.~Cartea and M.~Figueroa.
\newblock Pricing in {E}lectricity {M}arkets: a {M}ean {R}everting {J}ump
  {D}iffusion {M}odel with {S}easonality.
\newblock {\em Applied Mathematical Finance, No. 4, December 2005},
  12(4):313--335, 2005.

\bibitem{ContTankov2004}
R.~Cont and P.~Tankov.
\newblock {\em Financial {M}odelling with {J}ump {P}rocesses}.
\newblock Chapman and Hall, London, 2004.

\bibitem{Cufaro08}
N.~{Cufaro Petroni}.
\newblock Self-decomposability and {S}elf-similarity: a {C}oncise {P}rimer.
\newblock {\em Physica A, Statistical Mechanics and its Applications},
  387(7-9):1875--1894, 2008.

\bibitem{cs20_3}
N.~{Cufaro Petroni} and P.~Sabino.
\newblock Tempered {S}table {D}istribution and {F}inite {V}ariation
  {O}rnstein-{U}hlenbeck {P}rocesses.
\newblock Available at: https://arxiv.org/abs/2011.09147.

\bibitem{CKM17}
M.~Cummins, G.~Kiely, and B.~Murphy.
\newblock Gas {S}torage {V}aluation under {L}\'{e}vy {P}rocesses using {F}ast
  {F}ourier {T}ransform.
\newblock {\em Journal of Energy Markets}, 4:43--86, 2017.

\bibitem{CKM18}
M.~Cummins, G.~Kiely, and B.~Murphy.
\newblock Gas {S}torage {V}aluation under {M}ultifactor {L}\'{e}vy {P}rocesses.
\newblock {\em Journal of Banking and Finance}, 95:167--184, 2018.

\bibitem{DJS09}
C.~de~Jong and S.~Schneider.
\newblock Cointegration between {G}as and {P}ower {S}pot {P}rices.
\newblock {\em The Journal of Energy Markets}, 2(3):27--46, 2009.

\bibitem{Dev86}
L.~Devroye.
\newblock {\em Non-{U}niform {R}andom {V}ariate {G}eneration}.
\newblock Springer-Verlag, New York, 1986.

\bibitem{Dev2009}
L.~Devroye.
\newblock Random {V}ariate {G}eneration for {E}xponential and {P}olynomially
  {T}ilted {S}table {D}istributions.
\newblock {\em ACM Transactions on Modeling and Computer Simulation}, 19(4),
  2009.
\newblock Article No. 18.

\bibitem{Gardini20b}
M.~Gardini, P.~Sabino, and E.~Sasso.
\newblock A {B}ivariate {N}ormal {I}nverse {G}aussian {P}rocess with
  {S}tochastic {D}elay: {E}fficient {S}imulations and {A}pplications to
  {E}nergy {M}arkets, 2020.
\newblock Available at www.arxiv.org.

\bibitem{Gardini20a}
M.~Gardini, P.~Sabino, and E.~Sasso.
\newblock Correlating {L}\'evy {P}rocesses with {S}elf-decomposability:
  {A}pplications to {E}nergy {M}arkets, 2020.
\newblock Available at www.arxiv.org.

\bibitem{Grabchak16}
M.~Grabchak.
\newblock {\em Tempered {S}table {D}istributions}.
\newblock Springer International Publishing, 2016.

\bibitem{gradshteyn2007}
I.~S. Gradshteyn and I.~M. Ryzhik.
\newblock {\em Table of {I}ntegrals, {S}eries, and {P}roducts}.
\newblock Elsevier/Academic Press, Amsterdam, seventh edition, 2007.

\bibitem{HHM11}
B.~Hambly, S.~Howison, and T.~Kluge.
\newblock Information-{B}ased {M}odels for {F}inance and {I}nsurance.
\newblock {\em Quantitative Finance}, 9(8):937--949, 2009.

\bibitem{HHK09}
B.~Hambly, S.~Howison, and T.~Kluge.
\newblock Modelling {S}pikes and {P}ricing {S}wing {O}ptions in {E}lectricity
  {M}arkets.
\newblock {\em Quantitative Finance}, 9(8):937--949, 2009.

\bibitem{Hofert2012}
M.~Hofert.
\newblock Sampling {E}xponentially {T}ilted {S}table {D}istributions.
\newblock {\em ACM Transactions on Modeling and Computer Simulation}, 22(1),
  2012.

\bibitem{JL16}
E.~Jaeck and D.~Lautier.
\newblock Volatility in {E}lectricity {D}erivative {M}arkets: The {S}amuelson
  {E}ffect {R}evisited.
\newblock {\em Energy Economics}, 59:300--313, 2016.

\bibitem{JRT04}
P.~Jaillet, E.I. Ronn, and S.~Tompaidis.
\newblock Valuation of {C}ommodity-{B}ased {S}wing {O}ptions.
\newblock {\em Management Science}, 50(7):909--921, 2004.

\bibitem{KSB09}
R.~Kiesel, G.~Schindlmayr, and R.H. Börger.
\newblock A {T}wo-factor {M}odel for the {E}lectricity {F}orward {M}arket.
\newblock {\em Quantitative Finance}, 9(3):279--287, 2009.

\bibitem{Kjaer2008}
M.~Kjaer.
\newblock Pricing of {S}wing {O}ptions in a {M}ean {R}everting {M}odel with
  {J}umps.
\newblock {\em Applied Mathematical Finance}, 15(5-6):479--502, 2008.

\bibitem{LPV19}
L.~Latini, M.~Piccirilli, and T.~Vargiolu.
\newblock Mean-reverting {N}o-arbitrage {A}dditive {M}odels for {F}orward
  {C}urves in {E}nergy {M}arkets.
\newblock {\em Energy Economics}, 79:157--170, 2019.

\bibitem{LSW01}
F.~A. Longstaff and E.S. Schwartz.
\newblock Valuing {A}merican {O}ptions by {S}imulation: a {S}imple
  {L}east-{S}quares {A}pproach.
\newblock {\em Review of Financial Studies}, 14(1):113--147, 2001.

\bibitem{lu2020}
K.W. Lu.
\newblock Calibration for {M}ultivariate {L}\'evy-{D}riven
  {O}rnstein-{U}hlenbeck {P}rocesses with {A}pplications to {W}eak
  {S}ubordination, 2020.
\newblock Available at www.arxiv.org.

\bibitem{LS02}
J.J. Lucia and E.S. Schwartz.
\newblock Electricity {P}rices and {P}ower {D}erivatives: {E}vidence from the
  {N}ordic {P}ower {E}xchange.
\newblock {\em Review of Derivatives Research}, 5(1):5--50, Jan 2002.

\bibitem{SL2010}
E.~Luciano and P.~Semeraro.
\newblock {Multivariate {T}ime {C}hanges for {L}\'{e}vy {A}sset {M}odels:
  {C}haracterization and {C}alibration}.
\newblock {\em Journal of Computational and Applied Mathematics},
  233(1):1937--1953, 2010.

\bibitem{MBT2008}
T.~Meyer-Brandis and P.~Tankov.
\newblock Multi-factor {J}ump-diffusion {M}odels of {E}lectricity {P}rices.
\newblock {\em International Journal of Theoretical and Applied Finance},
  11(5):503--528, 2008.

\bibitem{MSH76}
J.~R. Michael, W.~R. Schucany, and R.~W. Haas.
\newblock Generating {R}andom {V}ariates {U}sing {T}ransformations with
  {M}ultiple {R}oots.
\newblock {\em The American Statistician}, 30(2):88--90, 1976.

\bibitem{PellegrinoSabino15}
T.~Pellegrino and P.~Sabino.
\newblock Enhancing {L}east {S}quares {M}onte {C}arlo with {D}iffusion
  {B}ridges: an {A}pplication to {E}nergy {F}acilities.
\newblock {\em Quantitative Finance}, 15(5):761--772, 2015.

\bibitem{PSV20}
M.~Piccirilli, M.D. Schmeck, and T.~Vargiolu.
\newblock Capturing the {P}ower {O}ptions {S}mile by an {A}dditive {T}wo-factor
  {M}odel for {O}verlapping {F}utures {P}rices.
\newblock {\em Energy Economics}, 95:105006, 2021.

\bibitem{ROSINSKI2007677}
Jan Rosinski.
\newblock Tempering {S}table {P}roceses.
\newblock {\em Stochastic Processes and their Applications}, 117(6):677 -- 707,
  2007.

\bibitem{Sabino21a}
P.~Sabino.
\newblock Pricing {E}nergy {D}erivatives in {M}arkets {D}riven by {T}empered
  {S}table and {CGMY} {P}rocesses of {O}rnstein-{U}hlenbeck {T}ype.
\newblock available at www.arxiv.org.

\bibitem{Sabino20b}
P.~Sabino.
\newblock Exact {S}imulation of {V}ariance {G}amma {R}elated {OU} {P}roceses:
  {A}pplication to the {P}ricing of {E}nergy {D}erivatives.
\newblock {\em Applied Mathematical Finance}, 27(3):207--227, 2020.

\bibitem{Sabino20a}
P.~Sabino.
\newblock Forward or {B}ackward {S}imulations? {A} {C}omparative {S}tudy.
\newblock {\em Quantitative Finance}, 20(7):1213--1226, 2020.

\bibitem{cs20_2}
P.~Sabino and N.~Cufaro Petroni.
\newblock Fast {P}ricing of {E}nergy {D}erivatives with {M}ean-{R}everting
  {J}ump-diffusion {P}rocesses.
\newblock {\em Applied Mathematical Finance}, 0(0):1--22, 2021.

\bibitem{Sato}
K.\ Sato.
\newblock {\em L\'evy {P}rocesses and {I}nfinitely {D}ivisible
  {D}istributions}.
\newblock Cambridge U.P., Cambridge, 1999.

\bibitem{SchwSchm00}
P.~Schwartz and J.E. Smith.
\newblock Short-term {V}ariations and {L}ong-term {D}ynamics in {C}ommodity
  {P}rices.
\newblock {\em Management Science}, 46(7):893--911, 2000.

\bibitem{ZhangOosterlee13_b}
B.~Zhang and C.W. Oosterlee.
\newblock An {E}fficient {P}ricing {A}lgorithm for {S}wing {O}ptions based on
  {F}ourier {C}osine {E}xpansions.
\newblock {\em Journal of Computational Finance}, 16(4):1--32, 2013.

\bibitem{ZhangOosterlee13_a}
B.~Zhang and C.W. Oosterlee.
\newblock Efficient {P}ricing of {E}uropean-style {A}sian {O}ptions under
  {E}xponential {L}{\'{e}}vy {P}rocesses based on {F}ourier {C}osine
  {E}xpansions.
\newblock {\em {SIAM} J. Financial Math.}, 4(1):399--426, 2013.

\end{thebibliography}
\end{document}